\documentclass[sigconf, screen=true,nonacm=true]{acmart}
\usepackage[english]{babel}
\usepackage[utf8x]{inputenc}

\usepackage{amsmath,amssymb,amsfonts,mathrsfs}
\usepackage{tkz-graph}
\usepackage{tikz}
\usepackage{mathtools}
\usepackage{pifont}
\usepackage{cleveref}
\usepackage{booktabs}
\usetikzlibrary{arrows,patterns, positioning}

\usepackage{enumitem}
\usepackage{xcolor}
\usepackage{balance}

\usepackage[outercaption]{sidecap}

\newsavebox\MBox

\def\mathunderline#1#2{\color{#1}\underline{{\color{black}#2}}\color{black}}

\usetikzlibrary{shapes,calc,arrows}
\usepackage[font={bf,sf,small}]{caption}
\definecolor{turquoise}{rgb}{0.1,0.6,0.3}
\definecolor{darkturquoise}{rgb}{0.2,0.5,0.4}
\definecolor{lightturquoise}{rgb}{0.1,0.8,0.2}
\definecolor{intturquoise}{rgb}{0.1,0.7,0.25}

\definecolor{hedge_1}{rgb}{0.9, 0.05, 0.05}
\definecolor{hedge_2}{rgb}{0.1, 0.7, 0.3}
\definecolor{hedge_3}{rgb}{0.1, 0.1, 0.9}
\definecolor{hedge_4}{rgb}{1, 0.6, 0.05}
\definecolor{hedge_5}{rgb}{0.8, 0.1, 0.8}

\definecolor{bough_1}{rgb}{0.9, 0.1, 0.1}
\definecolor{bough_2}{rgb}{1.0, 0.6, 0}
\definecolor{bough_3}{rgb}{0.3, 0.4, 0.95}
\definecolor{bough_4}{rgb}{1, 0.6, 0.05}
\definecolor{bough_5}{rgb}{0.8, 0.1, 0.8}
\definecolor{bough_6}{rgb}{1, 1, 0.05}

\colorlet{pastel_1}{bough_1!70}
\colorlet{pastel_2}{bough_2!80}
\colorlet{pastel_3}{bough_3!80}
\colorlet{pastel_4}{bough_4!80}
\colorlet{pastel_5}{bough_5!50}
\colorlet{pastel_6}{bough_6!40}

\colorlet{bwsafe_1}{pastel_1}
\colorlet{bwsafe_2}{pastel_2!90}
\colorlet{bwsafe_3}{pastel_3!90!black}
\colorlet{bwsafe_4}{pastel_5!90!black}

\colorlet{neutral-highlight}{black!10}

\definecolor{mygray}{rgb}{0.8,0.8,0.8}

\definecolor{highlight-1}{rgb}{0.8,0.15,0.15}

\usetikzlibrary{shapes}
\usetikzlibrary{plotmarks}
 \usetikzlibrary{calc,fit}

\newtheorem{theorem}{Theorem}
\newtheorem{observation}[theorem]{Observation}

\begin{document}
\fancyhead{}
\fancyfoot[C]{\thepage}
\setcopyright{none}


\author{Barbara Geissmann}
\affiliation{%
	\institution{Department of Computer Science, ETH Zurich}
	\city{Zurich}
	\state{Switzerland}
}

\author{Lukas Gianinazzi}
\affiliation{%
	\institution{Department of Computer Science, ETH Zurich}
	\city{Zurich}
	\state{Switzerland}
}
\email{lukas.gianinazzi@inf.ethz.ch}


\title{Parallel Minimum Cuts in Near-linear Work and Low Depth}

\begin{abstract}
	We present the first \emph{near-linear work} and \emph{poly-logarithmic depth} algorithm for computing a \emph{minimum cut}
in a graph, while previous parallel algorithms with poly-logarithmic depth required at least quadratic work in the number of vertices. 
	
	In a graph with $n$ vertices and $m$ edges, our algorithm computes the correct result with high probability in 
$O(m \log^4 n)$ work and $O(\log^3 n)$ depth.	
	This result is obtained by parallelizing a data structure that aggregates weights along paths in a tree and by exploiting the connection 
between minimum cuts and approximate maximum packings of spanning trees.
	
	In addition, our algorithm improves upon bounds on the number of cache misses incurred to compute a minimum cut.
\end{abstract}

\maketitle

\section{Introduction}

Two trends have emerged in microprocessor design in the last two decades:
(1) larger caches allow fast access to recently used memory locations and (2) many processing elements can be placed on the same chip, allowing for massively parallel processing. This has led to interest in both algorithms that take caches into account and parallel algorithm in a variety of different settings.

We consider \emph{shared-memory parallel} algorithms for computing a \emph{minimum cut} -- a fundamental subject in graph theory, that has many applications in practice, such as in network reliability~\cite{DBLP:journals/siamcomp/Karger99} and cluster analysis~\cite{DBLP:conf/sigir/Botafogo93,DBLP:journals/ipl/HartuvS00,DBLP:conf/ismb/SharanS00}. Our algorithm is based on one of the fastest known minimum cut algorithms, Karger's algorithm~\cite{DBLP:journals/siamcomp/Karger99}. It exploits a random edge-sampling technique and returns the correct result with high probability. Recently, we presented a cache-efficient variant of that algorithm~\cite{DBLP:conf/ciac/GeissmannG17}. Now, we build on that result by parallelizing a key data structure in the algorithm and obtain a parallel minimum cut algorithm that has low overhead compared to the sequential one (only logarithmic in the number of vertices).

We identify two main challenges in parallelizing graph algorithms. The first challenge is how to parallelize graph searches, since traversing a graph in parallel is problematic, especially when the graph has large diameter. The second challenge is that many graph algorithms (including those for minimum cuts~\cite{karger2000minimum,DBLP:journals/jacm/StoerW97}) employ intricate data structures for good performance. 
This is also problematic, because repeatedly accessing a data structure creates a sequential bottle-neck or can lead to concurrent accesses.

Our parallel minimum cut algorithm solves the first challenge by computing spanning trees that determine the order in which the edges of the input graph are accessed. In contrast to graphs, spanning trees can be traversed efficiently. Additionally, using parallel sorting, we rearrange the edges of the input graph to the order dictated by the traversal of the spanning tree, which avoids having to naively search the graph.

To solve the second challenge, we perform many data structure operations at once and in parallel. This works because the control flow of our algorithm does not depend on the result of the data structure operations until the very end, when the results from all data structure operations is aggregated efficiently in parallel. 

\subsection{Preliminaries}

\subsubsection{Graphs}
We consider an undirected weighted graph $G$ with vertices $V$, edges $E$, and positive edge weights $w\colon E \mapsto \mathbb{N}^{+}$. The number of vertices $|V|$ is $n$ and the number of edges $|E|$ is $m$. 

A nonempty proper subset of the vertices $V$ is a \emph{cut} $C$ of the graph $G$. A cut $C$ induces a partition of the vertices into two nonempty sets $C$ and $\bar C = V-C$. An edge $\{u, v\}$ that has endpoints in different parts of the partition ($u \in C$ and $v \in \bar C$) \emph{crosses} the cut $C$: it is a \emph{crossing edge}.
The total weight of the crossing edges of a cut $C$ is the {\em value} of $C$. A cut of smallest value is a \emph{minimum cut}. In particular, a disconnected graph has a minimum cut of value $0$.

\begin{figure}[b]
	\resizebox{0.8\linewidth}{!} {
\begin{tikzpicture}


\GraphInit[vstyle=hasse]
 \SetUpEdge[lw         = 1.2pt,
            color      = black,
            labelcolor = white,
            labeltext  = black,
            labelstyle = {fill=white}]

 \Vertex[x=2, y=1]{G}
 \Vertex[x=3, y=3]{A} 
 \Vertex[x=6, y=1]{P}
 \Vertex[x=4, y=1]{C}
 \Vertex[x=7, y=3]{Q}
 \Vertex[x=8, y=1]{E}
 \Edges[label=1](Q, E)
 \AddVertexColor{black}{A,G,C}
 \Edges[label=3](G,C) \Edges[label=3](A,G,C) \Edges[label=2](P,E)
 \Edges[label=1](Q,P) \Edges[label=1,style=dashed,color=red](Q, A) \Edges[label=1,style=dashed,color=red](C, P) \Edges[label=2](C,A)
\end{tikzpicture}
}
\caption{The vertex shading indicates the vertex partition of the minimum cut, which has value $2$. The crossing edges are dashed.}
\end{figure}
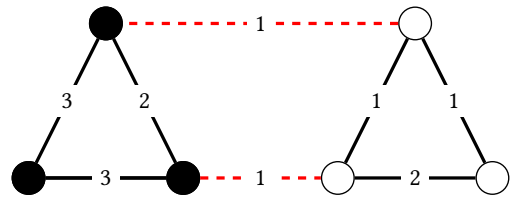

In this work, we will also consider directed trees. For two vertices $u$ and $v$ in a directed tree, 
we say that $v$ is a \emph{descendant} of $u$ (and $u$ is an \emph{ancestor} of $v$) if there is a directed path from $u$ to $v$. In particular, every vertex is its own descendant and its own ancestor. We shall denote the set of all descendants of $v$ as $v^{\downarrow}$ and the set of all ancestors of $v$ as $v^{\uparrow}$.

\subsubsection{Model of Computation}

The well-known parallel random access machine (PRAM) model~\cite{Reif:1993:SPA:562546} consists of a set of $p$ processors each connected to an unbounded \emph{shared memory}. This shared memory is organized into word-sized addressable locations. In each time step, every processor can read $O(1)$ memory locations and perform an $O(1)$ computable function on those words (this includes basic arithmetic, logic, control-flow, and addressing computations). Then, each processor can write back into $O(1)$ locations in the shared memory. The processors are \emph{synchronous}, which means that they proceed at the same speed and complete a time step together.

The \emph{runtime} of a PRAM algorithm is the number of time steps until the result is available in the shared memory. The runtime is determined by the processor that needs the most time steps.

We allow for multiple processors to read from the same memory location in the same time step, but forbid writing to the same location in the same time step. This is called the \emph{concurrent-read exclusive-write} (CREW) PRAM setting.

The Work-Depth model~\cite{Blelloch:1996:PPA:227234.227246} abstracts further from a concrete machine. In particular, a computation is viewed as a \emph{directed acyclic graph} (DAG), where each node in the graph corresponds to a constant time operation. The out-edges of a node correspond to the outputs of the operation executed in this node, and its $O(1)$ in-edges correspond to the inputs to the operation. Consequently, the input of an algorithm is given at a designated set of nodes, and the output has to be available at another set of designated nodes. 

The \emph{work} of an algorithm is its number of nodes in the computation DAG (not counting the input nodes). The \emph{depth} of an algorithm is the longest path from an input node to an output node.

Observe that the Work-Depth model and the PRAM model are closely related, as every PRAM algorithm can be viewed as generating a computation DAG. Conversely, the work and depth bounds translate to PRAM bounds. An algorithm with work $W$ and depth $D$ takes $O(W/p + D)$ time using $p$ processors in CREW PRAM~\cite{DBLP:journals/jacm/Brent74, Blelloch:1996:PPA:227234.227246}. 

\subsubsection{Randomization} 
We assume that in each time step, each processor has access to a uniformly random and independent bit. We distinguish between two types of randomized algorithms:
A \textit{Monte Carlo} randomized algorithm returns the correct result \emph{with high probability}. This means that the probability to return the wrong result can be made smaller than $1/n^c$ for any constant $c$. In particular, increasing $c$ by a constant factor only changes the runtime by a constant factor.
A \textit{Las Vegas} randomized algorithm always returns the correct result, but the runtime bounds are probabilistic and hold with high probability.

\subsection{Related Work}

\subsubsection{Relation to Maximum Flow} The minimum cut problem is a variant of the minimum $s$-$t$-cut problem, where the two designated vertices $s$ and $t$ must be in different parts of the partition.

The well-known Minimum-Cut-Maximum-Flow theorem~\cite{papadimitriou1998combinatorial} says that the value of a minimum $s$-$t$ cut equals the value of a maximum $s$-$t$ flow in the same network. Many maximum $s$-$t$ flow algorithms exist~\cite{DBLP:journals/jacm/GoldbergT88,DBLP:books/el/leeuwen90/KarpR90,dinic1970algorithm}, the best of which obtain $O(mn + n^2 \log n)$ runtime in general~\cite{DBLP:journals/jacm/GoldbergT88}.

A maximum $s$-$t$ flow can be computed with $O(n^2 \log n)$ depth~\cite{DBLP:journals/jal/ShiloachV82a}, which provides some speedup for denser graphs. When the value of the maximum flow is small, better bounds are obtained~\cite{DBLP:books/el/leeuwen90/KarpR90}.
\subsubsection{Deterministic Minimum Cut} 
A minimum cut can be computed by fixing an arbitrary vertex $s$ and computing a maximum $s$-$t$ flow for all different vertices $t\ne s$. In general, such an approach leads to work $\Omega(m n^2)$ using the known algorithms. However, the ideas from computing maximum flows can be adapted~\cite{DBLP:journals/jal/HaoO94} to a sequential algorithm with work $O(mn \log (n^2/m))$ .

Slightly better bounds of $O(mn + n^2 \log n)$ are obtained by a relatively simple approach based on a graph search~\cite{DBLP:journals/jacm/StoerW97}. This approach is a simplification of an approach by Nagomochi and Ibaraki~\cite{DBLP:conf/sigal/NagamochiI90}.

For \emph{unweighted} graphs, a recent result~\cite{DBLP:conf/stoc/KawarabayashiT15} obtains near-linear work for a deterministic (sequential) minimum cut algorithm.

\subsubsection{Randomized Minimum Cut} 
Randomized algorithms obtain both better work and depth than deterministic algorithms. Karger and Stein~\cite{karger1996new} give a Monte Carlo algorithm with  $O(n^2 \log ^ 3 n)$  work and $O(\log ^ 3 n)$ depth, which is faster than any known maximum flow algorithm when $m=\Omega(n \log^3 n)$.
Karger's algorithm~\cite{karger2000minimum} has the best known sequential bounds: it takes $O(m \log ^3 n)$ work. However, the parallel variant of that algorithm uses $O(n^2 \log n)$ work to obtain $O(\log ^ 3 n)$ depth.
\subsection{Our Contributions} 

\renewcommand{\arraystretch}{1.2}
\begin{table}[b]
	\vspace{-1em}	
	\centering
	\small
	\tabcolsep=3mm
	\begin{tabular}{lccc} %
		\toprule
		& Work & Depth \\
		\midrule
		Lowest Work~\cite{karger2000minimum} & $\Theta(m \log^3 n)$ & $\Theta(m \log n)$  \\
		
		Best Previous Polylog-Depth ~\cite{karger2000minimum} & $\Theta(n^2 \log n)$ & $\Theta(\log^3 n)$ \\
		
		\emph{This Paper} & $O(m \log^4 n)$ & $O(\log^3 n)$  \\
	\end{tabular}
	\vspace{1em}
\caption{Bounds for Computing a Minimum Cut. All algorithms are randomized and return correct results with high probability.}
\label{tab:bounds}
\vspace{-1em}
\end{table}
Our main contribution is a randomized parallel minimum cut algorithm that has near-linear work $O(m \log^4 n)$ and low depth $O(\log ^ 3 n)$. It returns the correct result with high probability. Previous parallel algorithms~\cite{karger2000minimum,karger1996new} with poly-logarithmic depth have quadratic work $\Omega(n^2 \log n)$ in the number of vertices. Our new algorithm, presented in \Cref{sec:parallel-mincuts}, is thus much more work-efficient when the graphs are not too dense, that is, $m \leq o(n^2 / \log ^ 3  n)$. \Cref{tab:bounds} compares our results to previous work.

As part of our solution, we present a parallel algorithm to solve a type of constrained minimum cut problem. Given a spanning tree $T$ of the graph, we find the cut of smallest value under the additional constraint that at most $2$ crossing edges are part of the spanning tree $T$. See \Cref{fig:constrained-cut} for an illustration of the problem. Our algorithm has work $O(m \log^2 n)$ and depth $O(\log ^2 n)$, with high probability. The best previous algorithm~\cite{karger2000minimum} with poly-logarithmic depth has $\Theta(n^2)$ work and $\Omega(\log n)$ depth.

We solve the constrained minimum cut problem by parallelizing a data structure that maintains aggregates of weights along paths in a tree. In this data structure, we consider a fixed tree where every vertex has a variable weight. The queries find the smallest weight in a path of the tree. The updates add a fixed weight to every vertex in a path of the tree (potentially changing many weights). In \Cref{sec:parallel-minpath}, we show how to answer a batch of $k$ mixed queries and updates on a tree of $n$ vertices with $O(k \log n (\log n + \log k) + n \log n)$ work and $O(\log^2 n + \log n \log k)$ depth. Hence, the average work per query and update is $O(\log^2 k)$, when $k \geq \Omega(n)$.

Our new approach also improves the number of cache misses to compute a minimum cut in the cache-oblivious model~\cite{Frigo99cache-obliviousalgorithms,Brodal}, where width of a cache line is $B$ and the size of the cache is $M$. As discussed in \Cref{sec:CO}, it incurs $O(\lceil (m \log^ 4 n) / B \rceil )$ cache misses and takes $O(m \log ^ 4 n)$ computation time. The best previous result~\cite{DBLP:conf/ciac/GeissmannG17} incurs $\Theta(\lceil { (m (\log ^ 4 n) \log _M n) / B} \rceil)$ cache misses and takes $\Theta(m \log ^ 5 n)$ computation time. Hence, the new algorithm improves the number of cache misses by a factor $\Theta(\log n / \log M)$, and the computation time by a factor $\Theta(\log n)$.

		\begin{figure}[t]
		 
			\begin{minipage}[t]{.49\linewidth}
			\centering
			\resizebox{\linewidth}{!} {
			\begin{tikzpicture}
			
			\definecolor{mygray}{rgb}{0.5,0.5,0.5}
			
			\GraphInit[vstyle=Hasse]
			\SetUpEdge[lw         = 1pt,
			color      = black,
			labelcolor = white,
			labeltext  = white,
			labelstyle = {sloped}]

			\begin{scope}
				\SetVertexNormal[FillColor  = black]
				\Vertex[x=-0.6, y=0.1, Math]{w_4}
				\Vertex[x=-1.55, y=-1, Math]{w_6}
				
			\end{scope}
			
			\begin{scope}
			 \SetVertexNormal[FillColor  = white]
			 \Vertex[x=1.2, y=-1, Math]{w_7}
			\Vertex[x=0.4, y=1.2, Math]{w_1}
			\Vertex[x=1.4, y=2.2, Math]{w_0}		
			\Vertex[x=2.8, y=-1, Math]{w_8}
			\Vertex[x=2, y=0.5, Math]{w_5}
			\Vertex[x=3.4, y=1.2, Math]{w_3}			
			\end{scope}

			\begin{scope}
			\tikzstyle{EdgeStyle}=[color=black]
				\Edges[lw=3pt](w_0,w_3)

					\Edges[lw=3pt](w_1, w_0)
					\Edges[lw=3pt](w_0, w_5)
					\Edges[lw=3pt](w_4, w_6)
					\Edges[lw=3pt](w_5, w_8)
					\Edges[lw=3pt](w_5, w_7)
			\end{scope}

			\begin{scope}
			\tikzstyle{EdgeStyle}=[color=black,style=dashed]
				\Edges[lw=3pt](w_1, w_4)
				
			\end{scope}
			
			\begin{scope}
			\tikzstyle{EdgeStyle}=[color=mygray,style=dashed]
				\Edges[](w_4, w_5)
				\Edges[](w_4, w_7)
				\Edges(w_6, w_7)
			\end{scope}
			
					\begin{scope}
			\tikzstyle{EdgeStyle}=[color=mygray]
				
				\Edges[](w_3, w_8)
				\Edges(w_7, w_8)
				
			\end{scope}

			\node[draw=none, fill=none] at (0.6, -2) {\huge One edge of the bold tree is cut.};

			\end{tikzpicture}
			}
		\end{minipage}%
		\hfill
  		\begin{minipage}[t]{.51\linewidth}
  			\centering
  			\resizebox{\linewidth}{!} {
			\begin{tikzpicture}
			
			\definecolor{mygray}{rgb}{0.5,0.5,0.5}
			
			\GraphInit[vstyle=Hasse]
			\SetUpEdge[lw         = 1pt,
			color      = black,
			labelcolor = white,
			labeltext  = white,
			labelstyle = {sloped}]

			\begin{scope}
				\SetVertexNormal[FillColor  = black]
				\Vertex[x=-0.6, y=0.1, Math]{w_4}
				\Vertex[x=-1.55, y=-1, Math]{w_6}
				\Vertex[x=1.2, y=-1, Math]{w_7}
				\Vertex[x=2.8, y=-1, Math]{w_8}
				\Vertex[x=2, y=0.5, Math]{w_5}
			\end{scope}
			
			\begin{scope}
			 \SetVertexNormal[FillColor  = white]
			\Vertex[x=0.4, y=1.2, Math]{w_1}
			\Vertex[x=1.4, y=2.2, Math]{w_0}		

			\Vertex[x=3.4, y=1.2, Math]{w_3}			
			\end{scope}

			\begin{scope}
			\tikzstyle{EdgeStyle}=[color=black]
				\Edges[lw=3pt](w_0,w_3)

					\Edges[lw=3pt](w_1, w_0)
					\Edges[lw=3pt](w_4, w_6)
					\Edges[lw=3pt](w_5, w_8)
					\Edges[lw=3pt](w_5, w_7)
			\end{scope}

			\begin{scope}
			\tikzstyle{EdgeStyle}=[color=black,style=dashed]
				\Edges[lw=3pt](w_1, w_4)
				\Edges[lw=3pt](w_0, w_5)
			\end{scope}
			
			\begin{scope}
			\tikzstyle{EdgeStyle}=[color=mygray,style=dashed]
				\Edges[](w_3, w_8)
			\end{scope}
			
					\begin{scope}
			\tikzstyle{EdgeStyle}=[color=mygray]
				\Edges[](w_4, w_7)
				\Edges(w_6, w_7)
				\Edges[](w_4, w_5)
				\Edges(w_7, w_8)
			\end{scope}

			\node[draw=none, fill=none] at (0.6, -2) {\huge Two edges of the bold tree are cut.};
			\end{tikzpicture}

			}
		\end{minipage}
			
			\caption{Example illustrating the constrained optimization problem. The cuts are illustrated by the vertex shading. The left cut cuts one edge of the spanning tree and the right one cuts two edges of the spanning tree. Cuts that cut more edges of the spanning tree do not have to be considered in the constrained optimization problem. }
			\label{fig:constrained-cut}
		\end{figure}

\section{Background}\label{sec:background}
\subsection{Karger's Minimum Cut Algorithm}\label{sec:kargers-algorithm}

On a high level, Karger's randomized algorithm~\cite{karger2000minimum} consists of two main steps:
\begin{enumerate}
	\item Find a set of spanning trees $S$ (with a special property as described below) in a graph $G$. Each of those trees gives rise to a more constrained optimization problem:
	\item For each spanning tree $T \in S$, compute the smallest cut $C$ of $G$ that has at most two crossing edges in $T$ (at most two edges of $T$ cross $C$). See \Cref{fig:constrained-cut} for an illustration.
\end{enumerate}

A key insight of Karger is a randomized procedure to find a set $S$ of only $O(\log n)$ spanning trees such that the smallest cut found in the second step is a minimum cut with high probability. 

The set $S$ is constructed using an approximate maximum packing procedure~\cite{DBLP:journals/mor/PlotkinST95}. This tree-packing procedure consists of sampling a sparse subgraph and performing a series of $O(\log ^2 n)$ minimum spanning tree computations. It can therefore be parallelized using known algorithms.

\begin{lemma}[Karger~\cite{karger2000minimum}] \label{lem:sampling}
	In $O(\log^3 n )$ time using $O(m + n \log n )$ processors, a set of spanning trees $S$ of a graph $G$ can be computed with the following properties:
	\begin{enumerate}[label=\alph*)]
		\item The set $S$ has size $O(\log n)$.
		\item With high probability, there is a tree $T\in S$ such that a minimum cut of $G$ cuts at most $2$ edges of $T$.
	\end{enumerate}
\end{lemma}

The sequential bottleneck of Karger's algorithm lies in its second step. Although finding the smallest cut that cuts exactly \emph{one edge} of a given spanning tree turns out to be relatively easy, finding the smallest cut that cuts exactly \emph{two edges} of a given spanning tree is challenging. Karger gives a $O(\log ^ 3 n)$ depth algorithm to do so, but it performs $ \Theta(n^2 \log n)$ work. And for their faster $O(m \log ^ 3 n)$ work algorithm no efficient parallelization is known.

\subsection{Minimum Path}

The problem left open by Karger is how to compute the smallest cut that cuts exactly \emph{two edges} of a given spanning tree in parallel. The \emph{sequential} $O(m \log ^ 3 n)$ work minimum cut algorithm builds a data structure called {\em Minimum Path} on each of the spanning trees obtained in the first step and maintains weights on the vertices of the spanning tree, which correspond to certain estimates of the minimum cut.

Given a rooted tree $T$ with $n$ vertices where each vertex $v$ has a weight $w(v)$, a \emph{Minimum Path} structure supports:
	\begin{itemize}
		\item \textsc{MinPath($v$)}: Returns the smallest weight of a vertex on the path from $v$ to the root of $T$. 
	\item \textsc{AddPath($v$, $x$)}: Adds $x$ to the weight of all vertices on the path  from $v$ to the root of $T$.
	\end{itemize}
See \Cref{fig:minpath-structure} for a illustration of the two operations.
	
This problem is a special case of dynamic trees~\cite{DBLP:journals/jcss/SleatorT83}, where $\Theta(\log n)$ time per query and update suffice. That this is optimal in the pointer machine setting follows from a recent lower bound on dynamic prefix sums~\cite{DBLP:journals/siamcomp/PatrascuD06}. The dynamic tree data structure~\cite{DBLP:journals/jcss/SleatorT83} is difficult to parallelize because it is based on tree rotations. Therefore, for our \emph{parallel} Minimum Path structure in \Cref{sec:parallel-minpath}, we will take a different approach.

\begin{figure}[t]
	
	\begin{minipage}[t]{.48\linewidth}
		
		\centering
		\resizebox{\linewidth}{!} {
			
			\begin{tikzpicture}
			
			\GraphInit[vstyle=Dijkstra]
			\SetUpEdge[lw         = 1pt,
			color      = black,
			labelcolor = white,
			labeltext  = black,
			labelstyle = {sloped}]
			
			\fill[bough_4!60,rotate=47] (1.1,0.5) ellipse (1.7 and 0.5);

			\begin{scope}
			\SetVertexNormal[FillColor  = neutral-highlight]
			\Vertex[x=0.4, y=1.2, Math]{w_1}
			\Vertex[x=-0.6, y=0.1, Math]{w_4}
			\Vertex[x=1.4, y=2.2, Math]{w_2}			
			\end{scope}
			
			\Vertex[x=-1.55, y=-1, Math]{w_6}	
			\Vertex[x=1.2, y=-1, Math]{w_7}
			\Vertex[x=2.8, y=-1, Math]{w_8}
			
			\Vertex[x=2, y=0.5, Math]{w_5}
			
			\Vertex[x=3.4, y=1.2, Math]{w_3}			
			
			\Edges[](w_0,w_3)
			
			\Edges[](w_1, w_2)
			\Edges(w_2, w_5, w_7)
			
			\Edges(w_1,w_4,w_6)
			
			\Edges(w_5, w_8)
			
			\node[draw=none, fill=none] at (-0.7, 	2.1) {\huge{\textsc{MinPath}($v_4$)}};
			
			\begin{scope}
			\tikzstyle{EdgeStyle}=[bend left=13]
			
			\end{scope}
			
			\end{tikzpicture}
		}
	\end{minipage}%
	\hfill
	\begin{minipage}[t]{.52\linewidth}
		\centering
		\resizebox{\linewidth}{!} {
			\begin{tikzpicture}
			\GraphInit[vstyle=Dijkstra]
			\SetUpEdge[lw         = 1pt,
			color      = black,
			labelcolor = white,
			labeltext  = black,
			labelstyle = {sloped}]
			
			\fill[bough_4!60,rotate=-65] (0.3,2.05) ellipse (2 and 0.65);
			
			\begin{scope}
			\SetVertexNormal[FillColor  = neutral-highlight]
			\Vertex[x=2.8, y=-1, Math]{w_8+x}
			\Vertex[x=2, y=0.5, Math]{w_5+x}
			\Vertex[x=1.4, y=2.2, Math]{w_2+x}			
			\end{scope}
			
			\Vertex[x=0.4, y=1.2, Math]{w_1}
			\Vertex[x=-0.6, y=0.1, Math]{w_4}
			\Vertex[x=-1.55, y=-1, Math]{w_6}
			\Vertex[x=1.2, y=-1, Math]{w_7}
			\Vertex[x=3.4, y=1.2, Math]{w_3}			
			
			\Edges[](w_2+x,w_3)
			\Edges[](w_1, w_2+x)
			\Edges(w_2+x, w_5+x, w_7)
			\Edges(w_1,w_4,w_6)
			\Edges(w_5+x, w_8+x)
			
			\node[draw=none, fill=none] at (-1, 2.1) {\huge{\textsc{AddPath}($v_8$, $x$)}};
			
			\begin{scope}
			\tikzstyle{EdgeStyle}=[bend left=13]
			
			\end{scope}
			
			\end{tikzpicture}
		}
	\end{minipage}
	\caption {Illustation of Minimum Path operations. Node $v_i$ stores value $w_i$. \emph{Left:} MinPath($v_4$) computes the minimum of the weights of the highlighted nodes. \emph{Right:} AddPath($v_8$, $x$) adds $x$ to the weight of the highlighted nodes.}
	\label{fig:minpath-structure}
\end{figure}
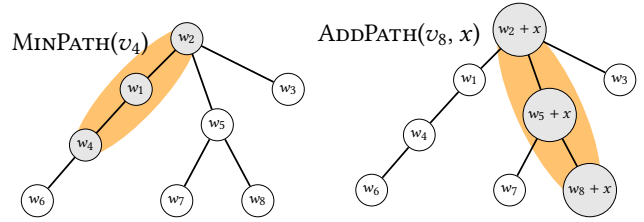

\subsection{Monotone Minimum Paths}\label{sec:monotone-minpaths}

We proceed with an overview of the data structure underlying the cache-oblivious algorithm minimum cut algorithm~\cite{DBLP:conf/ciac/GeissmannG17}. It is designed such that each operation accesses the memory in a monotone order (i.e., if an operation accesses location $x$ before location $y$ no operation accesses location $y$ before location $x$). This enables executing a batch of operations by sweeping through the memory only once, thus incurring a small number of cache misses. Similarly, this \emph{monotonicity} is also crucial for our parallel variant in \Cref{sec:parallel-minpath}.

\subsubsection{Tree Decomposition} The first step is to decompose the given tree $T$ into vertex-disjoint paths. To query or update a path $P$ in the tree $T$, simply perform an operation for each path in the decomposition that intersects the path $P$. 

It is possible to decompose the tree such that each root-to-leaf path is decomposed into $O(\log n)$ parts~\cite{DBLP:journals/jcss/SleatorT83,DBLP:conf/ciac/GeissmannG17}. \Cref{fig:decomp} illustrates such a decomposition and how a query on the tree corresponds to a set of queries on the paths in the decomposition. We present a parallel procedure to compute a tree decomposition in \Cref{sec:tree-decomposition}.

	\begin{figure}[t]
	\centering
	
	\begin{minipage}[t]{.5\linewidth}

			\resizebox{\linewidth}{!} {

			\begin{tikzpicture}

			\fill[bough_4!60,rotate=45] (1.1,0.5) ellipse (1.8 and 0.5);
			\fill[bough_4!60,rotate=100] (2.5,-1.8) ellipse (0.9 and 0.45);
			
			\GraphInit[vstyle=Dijkstra]
			
			\SetUpEdge[lw         = 1pt,
			color      = black,
			labelcolor = white,
			labeltext  = black,
			labelstyle = {sloped}]

			\begin{scope}
				\SetVertexNormal[FillColor  = pastel_1]
				\Vertex[x=-1.55, y=-1, Math]{w_6}			
			\end{scope}
			
			\begin{scope}
				\SetVertexNormal[FillColor  = pastel_1, LineWidth=0.4mm, LineColor=black!90]
				\Vertex[x=0.4, y=1.2, Math]{w_1}
				\Vertex[x=-0.6, y=0.1, Math]{w_4}			
			\end{scope}

			\begin{scope}
				\SetVertexNormal[FillColor  = pastel_2]
				\Vertex[x=1.2, y=0.1, Math]{w_7}
			\end{scope}
			
			\begin{scope}
				\SetVertexNormal[FillColor  = pastel_5]
				\Vertex[x=2.8, y=0.1, Math]{w_8}
				\Vertex[x=3.4, y=-1, Math]{w_9}
			\end{scope}

			\begin{scope}
				\SetVertexNormal[FillColor  = pastel_6, LineWidth=0.4mm, LineColor=black!90]
				\Vertex[x=1.2, y=3.4, Math]{w_r}	
				\Vertex[x=1.5, y=2.2, Math]{w_0}	
			\end{scope}
			
			\begin{scope}
				\SetVertexNormal[FillColor  = pastel_6]
				\Vertex[x=2, y=1.2, Math]{w_5}
			\end{scope}

			\begin{scope}
				\SetVertexNormal[FillColor  = pastel_3]
				\Vertex[x=3.4, y=1.2, Math]{w_3}			
			\end{scope}

			\Edges[style=dotted](w_0,w_3)

			\Edges[style=dotted](w_1, w_0)
			\Edges(w_r, w_0, w_5)
			\Edges[style=dotted](w_5, w_7)
			
			\Edges(w_1,w_4,w_6)
			
			\Edges(w_8, w_9)
				
			\Edges[style=dotted](w_5, w_8)

			\begin{scope}
			\tikzstyle{EdgeStyle}=[bend left=13]

			\end{scope}

			\end{tikzpicture}
		}
		\end{minipage}
		\hfill
		\begin{minipage}[t]{.47\linewidth}
		\resizebox{\linewidth}{!} {
		
			\begin{tikzpicture}
			
			\tikzstyle{myarrows}=[line width=1.5mm,draw=black!80,-triangle 45,postaction={draw, line width=3mm, shorten >=4mm, -}]

			\draw[myarrows] (-2.1,-0.5) -- (-0.8,-0.5);
			
			\fill[bough_4!60,rotate=0] (0.6,0) ellipse (0.8 and 0.5);
			\fill[bough_4!60,rotate=0] (0.6,1.3) ellipse (0.8 and 0.5);
			
			\GraphInit[vstyle=Dijkstra]

			\SetUpEdge[lw         = 1pt,
			color      = black,
			labelcolor = white,
			labeltext  = black,
			labelstyle = {sloped}]

			\begin{scope}
				\SetVertexNormal[FillColor  = pastel_1]
				\Vertex[x=2.4, y=0, Math]{w_6}			
			\end{scope}
			
			\begin{scope}
				\SetVertexNormal[FillColor  = pastel_6]
				\Vertex[x=2.4, y=1.3, Math]{w_5}
			\end{scope}
			
			\begin{scope}
				\SetVertexNormal[FillColor  = pastel_1, LineWidth=0.4mm, LineColor=black!90]
				\Vertex[x=0, y=0, Math]{w_1}
				\Vertex[x=1.2, y=0, Math]{w_4}			
			\end{scope}
			
			\begin{scope}
				\SetVertexNormal[FillColor  = pastel_6, LineWidth=0.4mm, LineColor=black!90]
				\Vertex[x=0, y=1.3, Math]{w_r}	
				\Vertex[x=1.2, y=1.3, Math]{w_0}
			\end{scope}

			\begin{scope}
				\SetVertexNormal[FillColor  = pastel_3]
				\Vertex[x=0, y=-1.1, Math]{w_3}			
			\end{scope}
			
			\begin{scope}
				\SetVertexNormal[FillColor  = pastel_2]
				\Vertex[x=0, y=-3, Math]{w_7}
			\end{scope}
			
			\begin{scope}
				\SetVertexNormal[FillColor  = pastel_5]
				\Vertex[x=0, y=-2, Math]{w_8}
				\Vertex[x=1.2, y=-2, Math]{w_9}
			\end{scope}

			\Edges(w_r, w_0, w_5)
						
			\Edges(w_1,w_4,w_6)
			
			\Edges(w_8, w_9)
			
			\end{tikzpicture}
		}
		\end{minipage}		
		\caption{An operation in the tree (left) is decomposed into operations on the paths in the decomposition (right).}
		\label{fig:decomp}
		\vspace{-1em}
	\end{figure}
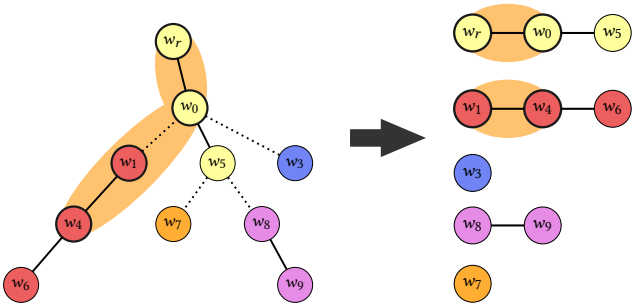

\subsubsection{List view}
We view each path of the decomposition as a list, where the vertex closest to the root is at the beginning (front) of the list. We call the deductive problem about querying and updating prefixes of lists \emph{Minimum Prefix}.
Specifically, we define \textit{MinPrefix} analogously to MinPath and \textit{AddPrefix} analogously to AddPath on a list $L$ of vertices $(v_1, \dotsc, v_n)$ with weights $(w_1, \dotsc, w_n)$:

\vspace{2em}
\noindent
\textsc{MinPrefix}($v_i$) returns the smallest weight in the prefix $w_1, \dotsc, w_i$. The example below illustrates MinPrefix($v_6$) on a list with $8$ nodes.
	
		\begin{center}		
  		\begin{minipage}[t]{.62\linewidth}
  			\centering
  			\resizebox{1.0\linewidth}{!} {
			\begin{tikzpicture}
				\definecolor{mygray}{rgb}{0.5,0.5,0.5}

			\GraphInit[vstyle=Dijkstra]
			\SetUpEdge[lw         = 1pt,
			color      = black,
			labelcolor = white,
			labeltext  = black,
			labelstyle = {sloped}]

			\fill[bough_4!60,rotate=0] (-0.9,-1.1) ellipse (2.6 and 0.55);
			
			\begin{scope}
				\SetVertexNormal[LineWidth=0.4mm, FillColor=neutral-highlight,LineColor=black!90]
				\Vertex[x=1.35, y=-1.1, Math]{w_6}
						\Vertex[x=-3.15, y=-1.1, Math]{w_1}
			\Vertex[x=-2.25, y=-1.1, Math]{w_2}
			\Vertex[x=-1.35, y=-1.1, Math]{w_3}
			\Vertex[x=-0.4, y=-1.1, Math]{w_4}
			\Vertex[x=0.5, y=-1.1, Math]{w_5}
			
			\end{scope}

			\Vertex[x=2.25, y=-1.1, Math]{w_7}
			\Vertex[x=3.15, y=-1.1, Math]{w_8}	
						
			\begin{scope}
			\tikzstyle{EdgeStyle}=[->, color=mygray]
			
			\Edges(w_8,w_7, w_6, w_5, w_4, w_3, w_2, w_1)

			\end{scope}

			\end{tikzpicture}
		}
	\end{minipage} \\
	\end{center}

	\noindent
	\textsc{AddPrefix}($v_i$, $x$) adds $x$ to the weights of the vertices in the prefix $v_1, \dotsc, v_i$. The example below illustrates AddPrefix($v_4$, $x$) on a list with $8$ nodes.
	
	\begin{center}
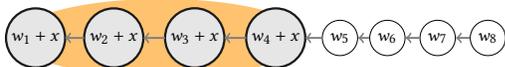
  	\begin{minipage}[t]{.8\linewidth}
	\resizebox{1.0\linewidth}{!} {
	
			\begin{tikzpicture}
			\definecolor{mygray}{rgb}{0.5,0.5,0.5}

			\GraphInit[vstyle=Dijkstra]
			\SetUpEdge[lw         = 1pt,
			color      = black,
			labelcolor = white,
			labeltext  = black,
			labelstyle = {sloped}]
		
			\fill[bough_4!60,rotate=0] (-2.7,-1.1) ellipse (2.6 and 0.7);			
			
			\begin{scope}
				\SetVertexNormal[FillColor  = neutral-highlight, LineWidth=0.4mm, LineColor=black!90]
				
			\Vertex[x=-4.95, y=-1.1, Math]{w_1+x}
			\Vertex[x=-3.55, y=-1.1, Math]{w_2+x}
			\Vertex[x=-2.1, y=-1.1, Math]{w_3+x}
			\Vertex[x=-0.65, y=-1.1, Math]{w_4+x}

			\end{scope}
			\Vertex[x=1.35, y=-1.1, Math]{w_6}
			\Vertex[x=0.5, y=-1.1, Math]{w_5}
			\Vertex[x=2.25, y=-1.1, Math]{w_7}
			\Vertex[x=3.15, y=-1.1, Math]{w_8}
						
			\begin{scope}
			\tikzstyle{EdgeStyle}=[->, color=mygray]
			
			\Edges(w_8,w_7, w_6, w_5, w_4+x, w_3+x, w_2+x, w_1+x)

			\end{scope}

			\end{tikzpicture}
		}		
	\end{minipage}
	\end{center}

To perform Minimum Prefix operations on a list $l$ we build a complete binary tree $B$ on top the vertices $v_1, \dotsc, v_n$, such that the vertices form the leaves of the tree $B$. This tree holds auxiliary information that allows us to perform the prefix operations quickly. A naive attempt would be to store in each inner node of the tree $B$ the minimum value in its subtree. This approach answers queries in $O(\log n)$ time, but an update takes $\Omega(n)$ time (for example when the prefix covers the whole list).

A better approach is to store only differences of minima: Each inner node stores the difference between the smallest leaf in its right subtree and the smallest leaf in its left subtree. With this approach, only $O(\log n)$ values change when the list is updated, namely those on the path from the root of the binary tree to the leaf corresponding to the last vertex that is updated. See \Cref{fig:update} for an illustration.

Moreover, these differences (more specifically their signs) suffice to determine which subtree, and consequently, which vertex of the list contains the smallest value.

In the following, we we describe how to efficiently maintain these difference values and how to use them to compute the smallest weight in a given prefix of the list.

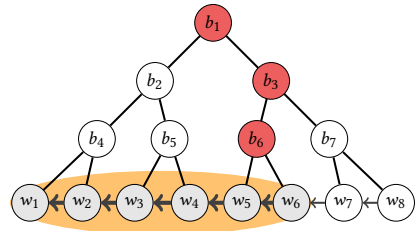
\begin{figure}[t]
	\centering
	\resizebox{0.65\linewidth}{!} {
		\begin{tikzpicture}
		\GraphInit[vstyle=Dijkstra]
		\SetUpEdge[lw         = 1pt,
		color      = black,
		labelcolor = white,
		labeltext  = black,
		labelstyle = {sloped}]
		
		\Vertex[x=2, y=0, Math]{b_7}

		\fill[bough_4!60,rotate=0] (-0.9,-1.1) ellipse (2.6 and 0.55);

		\begin{scope}
		\SetVertexNormal[FillColor  = neutral-highlight]
		
		\Vertex[x=-3.15, y=-1.1, Math]{w_1}
		\Vertex[x=-2.25, y=-1.1, Math]{w_2}
		\Vertex[x=-1.35, y=-1.1, Math]{w_3}
		\Vertex[x=-0.4, y=-1.1, Math]{w_4}
		\Vertex[x=0.5, y=-1.1, Math]{w_5}
		\Vertex[x=1.35, y=-1.1, Math]{w_6}	
		\end{scope}
		
		\Vertex[x=-1, y=1, Math]{b_2}
		\Vertex[x=-2, y=0, Math]{b_4}
		\Vertex[x=-0.75, y=0, Math]{b_5}
		\Vertex[x=2.25, y=-1.1, Math]{w_7}
		\Vertex[x=3.15, y=-1.1, Math]{w_8}
		
		\begin{scope}
		\SetVertexNormal[FillColor  = pastel_1]
		\Vertex[x=1, y=1, Math]{b_3}
		\Vertex[x=0.75, y=0, Math]{b_6}
		\Vertex[x=0, y=2, Math]{b_1}	
		
		\end{scope}
		
		\Edges[](b_1, b_2)
		\Edges[](b_1, b_3)
		\Edges[](b_4, b_2, b_5)
		\Edges(b_3, b_6)
		\Edges(b_3, b_7)
		
		\Edges(w_1, b_4, w_2)
		\Edges(w_3, b_5, w_4)
		\Edges(w_5, b_6, w_6)
		\Edges(w_7, b_7, w_8)
		
		\begin{scope}
		\tikzstyle{EdgeStyle}=[->, color=darkgray]
		\Edges(w_8,w_7, w_6)
		\end{scope}
		
		\begin{scope}
		\tikzstyle{EdgeStyle}=[->, color=darkgray]
		\Edges[lw=2pt](w_6, w_5, w_4, w_3, w_2, w_1)
		\end{scope}
		\end{tikzpicture}

	}
	
	\caption{At each inner node $b_i$, store the difference between the smallest weight in $b_i$'s right and left subtree. An update changes the differences on a single root-to-leaf path. The example above illustrates an update to the first six elements in the list.}
	\label{fig:update}
\end{figure}

\subsubsection{AddPrefix}\label{sec:seqAddPrefix} For any node $b$ with right child $r$ and left child~$l$, let $\min_i(b)$ be the smallest weight of any descendant of $b$ after the $i$-th update and let $\min_0(b)$ be the smallest weight of any descendant of $b$ in the initial state. Recall that this value is not stored directly for efficiency reasons, but instead, the data structure stores in each node $b$ at time $i$ the value \[\Delta_i(b) := \text{min}_i(r)-\text{min}_i(l) \enspace .\]

Let \textsc{AddPrefix}($v, x$) be the $i$-th update and assume that we know by how much the update changes the minimum in the right subtree (i.e. $\min_i(r)-\min_{i-1}(r)$) and the minimum in the left subtree (i.e. $\min_i(l)-\min_{i-1}(l)$). Then, we can derive by how much the difference between the two subtree minima changes in this update. Therefore, we define $\phi_i(b)=\text{min}_i(b)-\text{min}_{i-1}(b)$. We have
\begin{align*}
\Delta_i(b) = \Delta_{i-1}(b) + \phi_i(r) - \phi_i(l) \ .
\end{align*}

Of course, it would be too expensive to compute $\Delta_i(b)$ by recursively computing the values $\phi_i$ of both children, since all descendants would be traversed.
However, we observe that at least one child of every node has a trivial value: If all descendant leaves of a node $b'$ are in the prefix of the list that is updated, then $\phi_i(b')=x$. If no descendant leaves of a node $u$ are in the prefix of the list that is updated, then $\phi_i(u)=0$. This means that the value of $\phi_i$ is only non-trivial for the nodes on the path from $v$ to the root (which are also the only nodes where $\Delta$ changes).
Hence, to perform the $i$-th update \textsc{AddPrefix}($v$, $x$), we walk up along this path in the tree starting in the leaf $v$ and compute $\phi_i(b)$ for each node $b$ along this path. For the leaf, this is $\phi_i(v)=x$. For an interior node $b$, \Cref{fig:addpath-same-tree} and  \Cref{fig:addpath-different-tree} illustrate how the value of $\phi_i(b)$ is computed. There are two additional symmetric cases.

\subsubsection{MinPrefix}
 The queries also proceed walking up the path from the last vertex in the prefix to the root. We consider the weights and state of the data structure at a fixed time and for simplicity, we omit the time subscripts.
 
  When computing \textsc{MinPrefix}($v_k$), it is tempting to directly compute for each node along this path the result of \textsc{MinPrefix}($v_k$) restricted to the current subtree (namely $\min_{v_i \in b^\downarrow, i\leq k} w_i )$. Unfortunately, this is not possible using only the value of $\Delta(b)$. Instead, we compute the difference $d(b)$ between this quantity and the minimum of the current subtree $\min(b)$: 
  \[d(b) := \left(\min_{v_i \in b^\downarrow, i\leq k} w_i \right) - \min(b) \ .\]
  
  Once this difference is known for the root, we get the result of \textsc{MinPath}($v_k$) by adding the overall minimum to this difference.
  See \Cref{fig:query-case-left} and \Cref{fig:query-case-right} for an illustration for how $d(b)$ is computed bottom-up. There is an additional case symmetric to \Cref{fig:query-case-left}.

\subsubsection{Running time} The runtime depends on the way we decompose the tree $T$ into paths. Each path in $T$ is decomposed into at most $O(\log n)$ lists. As we can perform \textsc{MinPrefix} and \textsc{AddPrefix} on each list in the decomposition in $O(\log n)$ time, the overall time to perform \textsc{AddPath} and \textsc{MinPath} is $O(\log^2 n)$.

\begin{figure}[t!]
	
	\centering
	\begin{minipage}[t]{.48\linewidth}
		\centering
		\resizebox{.84\linewidth}{!} {
			
			\begin{tikzpicture}
			
			\GraphInit[vstyle=Dijkstra]
			\GraphInit[vstyle=Normal]
			\SetUpEdge[lw         = 1pt,
			color      = black,
			labelcolor = white,
			labeltext  = black,
			labelstyle = {left=2pt, yshift=0.3cm, fill=none}]

			\Vertex[x=0, y=2,]{b}	
			\Vertex[x=-1, y=1, Math]{l}
			\Vertex[x=1, y=1, Math]{r}
			
			\Edges[lw=0.5pt, label=$\phi_i(l)$](b, l)
			\Edges[lw=0.5pt, label=$\phi_i(r)$, labelstyle ={right =2pt, yshift=0.3cm, fill=none}](b, r)
			
			\node[isosceles triangle,shape border rotate=90, minimum height=20mm, minimum width=12mm, fill=black!5, draw] at (1,-0.8) {};
			\node[isosceles triangle,shape border rotate=90, minimum height=20mm, minimum width=12mm, fill=black!5, draw] at (-1,-0.8) {};
			
			\draw [->, thick] (1.5,-1.7) -- (1.5,-1.35);
			\node[draw=none,fill=none] at (1.3,-1.9) {$\min_{i}(b)$} ;
			
			\draw [->, thick] (0.5,-1.7) -- (0.5,-1.35);
			\node[draw=none,fill=none] at (0,-1.9) {$\min_{i-1}(b)$} ;
			
			\node[draw=none,fill=none] at (0,2.8) {$\phi_i(b) = \phi_i(r)$} ;

			\end{tikzpicture}
		}

		\caption{AddPrefix (A): In case the minimum stays in the right subtree after the $i$-th update, the change in minimum at node $b$ is given directly by the change of minimum in the right subtree.}
		\label{fig:addpath-same-tree}
		
	\end{minipage}%
	\hfill
	\begin{minipage}[t]{.48\linewidth}
		\centering
		\resizebox{.84\linewidth}{!} {
			
			\begin{tikzpicture}

			\GraphInit[vstyle=Dijkstra]
			\GraphInit[vstyle=Normal]
			\SetUpEdge[lw         = 1pt,
			color      = black,
			labelcolor = white,
			labeltext  = black,
			labelstyle = {left=2pt, yshift=0.3cm, fill=none}]

			\Vertex[x=0, y=2]{$b$}	
			\Vertex[x=-1, y=1]{$l$}
			\Vertex[x=1, y=1]{$r$}
			
			\Edges[lw=0.5pt, label={ $\phi_i(l)$}](b, l)
			\Edges[lw=0.5pt, label={ $\phi_i(r)$}, labelstyle = {right=2pt, yshift=0.3cm, fill=none}](b, r)
			
			\node[isosceles triangle,shape border rotate=90, minimum height=20mm, minimum width=12mm, fill=black!5, draw] at (1,-0.8) {};
			\node[isosceles triangle,shape border rotate=90, minimum height=20mm, minimum width=12mm, fill=black!5, draw] at (-1,-0.8) {};
			
			\draw [->, thick] (1,-1.7) -- (1,-1.35);
			\node[draw=none,fill=none] at (1,-1.9) {$\min_i(b)$} ;
			
			\draw [->, thick] (-1,-1.7) -- (-1,-1.35);
			\node[draw=none,fill=none] at (-1,-1.9) {$\min_{i-1}(b)$} ;
			
			\node[draw=none,fill=none] at (0,2.8) {$\phi_i(b) =\phi_i(r) + \Delta_{i-1}(b)$} ;

			\end{tikzpicture}
		}
		\vfill
		\caption{AddPrefix (B): In case the minimum was in the left subtree before the $i$-th subtree, and is in the right subtree after the $i$-th update, the difference between the minimum of the two subtrees before the $i$-th update needs to be taken into account.}
		\label{fig:addpath-different-tree}
	\end{minipage}
	
\end{figure}

\section{Parallel Minimum Path}\label{sec:parallel-minpath}

The monotone Minimum Path structure from \Cref{sec:monotone-minpaths} can execute Queries and Updates one-by-one. Directly trying to parallelize such an approach leads to problems of \emph{concurrency}: When many updates try to change the same memory location concurrently, these conflicts need to be resolved. Some general techniques to do so are known in practice, such as locks or lock-free methods~\cite{DBLP:books/daglib/0020056}. However, these approaches essentially serialize accesses to the same memory location. Thus, locations that are accessed by many updates, such as the root in the minimum path structure, become sequential bottlenecks.

Therefore, we take a different approach. 
We start with two observations implied by cache-oblivious algorithm~\cite{DBLP:conf/ciac/GeissmannG17}:
\begin{enumerate}
	\item[(1)] The complete sequence of updates and queries is known upfront: it is enough to find a parallel algorithm to perform a \emph{batch} of minimum path operations.
	\item[(2)] The data structure operations traverse the memory in a fixed order (namely walking bottom-up in some trees). This allows us to perform all of the updates at once, simulating the execution of all updates at the same time by logically sweeping the trees bottom-up and producing for each memory location all its intermediate states at once.
\end{enumerate}

In the cache-oblivious algorithm, the batch of updates is simulated using a priority queue. This is inherently sequential. Moreover, it leads to a runtime of $\Theta(\log ^ 3 n)$ per minimum path operation.

In Sections \ref{sec:parallel-addprefix} and \ref{sec:parallel-minprefix} we instead provide an explicit schedule for performing a batch of updates and queries \emph{in parallel}, such that the overall work per minimum path operation is $\Theta(\log ^2 n)$.

\subsubsection{Problem Statement}
We consider a \emph{batch} of $k$ \textsc{MinPath} and \textsc{AddPath} operations $o_1, \dotsc, o_k$. Each \textsc{AddPath} operation is of the form $o_i = (i, v(i), x(i))$ for a time $i$, a vertex $v(i)$, and a weight $x(i)$. Similarly, each \textsc{MinPath} operation $o_i$ is of the form $o_i = (i, v(i))$ for a time $i$ and a vertex $v(i)$. The goal is to compute the result of all \textsc{MinPath} operations at once in parallel, as if the operation sequence where executed sequentially.

As in the sequential case, we decompose the tree into a set of directed paths (See \Cref{sec:tree-decomposition})  and build a minimum prefix structure for each of those paths . We start with the presentation of \textit{parallel} \textsc{AddPrefix} and \textit{parallel} \textsc{MinPrefix}.

\begin{figure}[t]
	
	\centering
	\begin{minipage}[t]{.48\linewidth}
		\centering
		\resizebox{0.84\linewidth}{!} {
			
			\begin{tikzpicture}
			
			\GraphInit[vstyle=Dijkstra]
			\GraphInit[vstyle=Normal]
			\SetUpEdge[lw         = 1pt,
			color      = black,
			labelcolor = white,
			labeltext  = black,
			labelstyle = {left=2pt, yshift=0.3cm, fill=none}]

			\Vertex[x=0, y=2, Math]{b}	
			\Vertex[x=-1, y=1, Math]{l}
			\Vertex[x=1, y=1, Math]{r}
			
			\Edges[lw=0.5pt, label=$d(l)$](b, l)
			\Edges[lw=0.5pt, labelstyle = {right=15pt, above=-2pt, fill=none}](b, r)
			
			\draw [-, line width=1.5mm, color=bough_4] (-1.75,-1.15) -- (-1.15,-1.15);
			
			\node[isosceles triangle,shape border rotate=90, minimum height=18mm, minimum width=13mm, fill=black!5, draw] at (1,-0.65) {};
			\node[isosceles triangle,shape border rotate=90, minimum height=18mm, minimum width=13mm, fill=black!5, draw] at (-1, -0.65) {};
			
			\draw [->, thick] (-1.15,-1.5) -- (-1.15,-1.15);
			\node[circle, draw=black,fill=bough_4,inner sep=1pt] at (-1.15,-1.7) {$v_k$} ;
			
			\draw [->, thick] (-0.5,-1.5) -- (-0.5,-1.15);
			\node[draw=none,fill=none] at (-0.2,-1.7) {min(b)} ;
			
			\node[draw=none,fill=none] at (0,2.8) {$d(b) := d(l)$} ;
			
			\end{tikzpicture}
		}
		
		\caption{MinPrefix (A). The easy case for MinPrefix($v_k)$ is when both the query vertex $v_k$ and the smallest weight in $b$'s subtree $\min(b)$ are in the left subtree $l$. Then, $d(b)$ is copied from the left child $l$.}
		\label{fig:query-case-left}
	\end{minipage}%
	\hfill
	\begin{minipage}[t]{.48\linewidth}
		\centering
		\resizebox{0.84\linewidth}{!} {

			\begin{tikzpicture}

			\GraphInit[vstyle=Dijkstra]
			\GraphInit[vstyle=Normal]
			\SetUpEdge[lw         = 1pt,
			color      = black,
			labelcolor = white,
			labeltext  = black,
			labelstyle = {left=2pt, yshift=0.3cm, fill=none}]
			
			\Vertex[x=0, y=2, Math]{b}	
			\Vertex[x=-1, y=1, Math]{l}
			\Vertex[x=1, y=1, Math]{r}
			
			\Edges[lw=0.5pt, label=$d(l)$](b, l)
			\Edges[lw=0.5pt, labelstyle = {right=15pt, above=-2pt, fill=none}](b, r)
			
			\draw [-, line width=1.5mm, color=bough_4] (-1.75,-1.15) -- (-0.75,-1.15);
			\draw [-, dotted, line width=1.5mm, color=bough_4] (-0.75,-1.15) -- (-0.25,-1.15);
			\draw [-, dotted, line width=1.5mm, color=bough_4] (0.25,-1.15) -- (1.15,-1.15);
			
			\node[isosceles triangle,shape border rotate=90, minimum height=18mm, minimum width=13mm, fill=black!5, draw] at (1,-0.65) {};
			\node[isosceles triangle,shape border rotate=90, minimum height=18mm, minimum width=13mm, fill=black!5, draw] at (-1, -0.65) {};
			
			\draw [->, thick] (-1.15,-1.5) -- (-1.15,-1.15);
			\node[circle, draw=black,fill=bough_4, align=center,inner sep=1pt] at (-1.15,-1.7) {$w^{\star}$} ;
			
			\draw [->, thick] (1.4,-1.5) -- (1.4,-1.15);
			\node[draw=none,fill=none] at (1.3,-1.7) {min(b)} ;
			
			\node[draw=none,fill=none] at (0,2.8) {$d(b) := d(l) - \Delta(b)$} ;
			
			\end{tikzpicture}
		}
		
		\caption{MinPrefix (B): The smallest weight in the query prefix $w^{\star}=\min_{v_i \in b^\downarrow, i\leq k} w_i$ is in the left subtree $l$, but the smallest weight in $b$'s subtree $\min(b)$ is in the right subtree $r$. The difference $\Delta(b)$ between the minimum in the right subtree and the minimum in the left subtree needs to be taken into account. }
		\label{fig:query-case-right}
	\end{minipage}
	
\end{figure}

\subsection{Parallel AddPrefix}\label{sec:parallel-addprefix}

We consider a batch of \textsc{AddPrefix} operations for a list of length $n$, where each operation $o_i$ is of the form $o_i = (i, v(i), x(i))$ for a time $i$, a vertex $v(i)$, and a weight $x(i)$. Conceptually, we build the same binary tree $B$ on top of the list as in the sequential case (see \Cref{sec:monotone-minpaths}). 

In order to allow for future queries, we need to produce all intermediate states of a particular node that arose when the updates were executed sequentially.
From inspecting the update equation for a particular node $b$ with left child $l$ and right child $r$,
\begin{align*}
\Delta_{i}(b) &= \Delta_{i-1}(b) + \phi_i(r)-\phi_i(l) \ ,
\intertext{
we make the following key observation that it telescopes to} \Delta_i(b) &= \Delta_0(b) + \left(\sum_{j=1}^{i} \phi_j (r) \right)- \left(\sum_{j=1}^{i} \phi_j (l) \right) \enspace .
\end{align*}

Therefore, given the all-prefix-sums for $\phi_{1}(r), \dotsc, \phi_k(r)$ and for $\phi_{1}(l), \dotsc , \phi_k(l)$, we can compute the values of $\Delta_1 (b), \dotsc, \Delta_k(b)$ in $O(1)$ depth and $O(k)$ work. Moreover, given those values, the values $\phi_1(b), \dotsc, \phi_k(b)$ each only depend on a constant number of already computed terms. Therefore, they can also be computed in $O(1)$ depth and $O(k)$ work. Unfortunately, this naive interpretation still needs far too many processors: $\Omega(k)$ work at every node $b$, and hence $\Omega(nk)$ work overall.

Fortunately, not every update is relevant to every node: 
$\Delta_i(b)$ only changes for an update $o_i=(i, v(i), x(i))$ if $v(i)$ is a descendant of $b$. We can therefore restrict our attention to the times $i$ of those operations:
Let $H(b)$ be the union over all times $i$ where the update is such that $v(i)$ is a descendant of $b$.

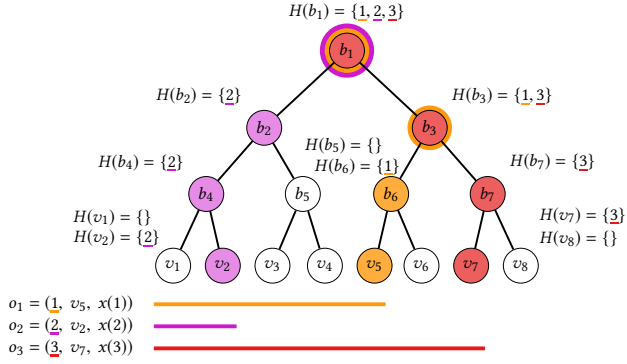
\begin{figure}[t]
	\centering
	\resizebox{1\linewidth}{!} {
		\begin{tikzpicture}[node distance = 0.1cm]
		\GraphInit[vstyle=Dijkstra]
		\SetUpEdge[lw         = 1pt,
		color      = black,
		labelcolor = white,
		labeltext  = black,
		labelstyle = {sloped}]

		\draw [-, line width=0.7mm, color=bough_4] (-3.5,-1.8) -- (0.7,-1.8);
		\draw [-, line width=0.7mm, color=bough_5] (-3.5,-2.2) -- (-2,-2.2);
		\draw [-, line width=0.7mm, color=bough_1] (-3.5,-2.6) -- (2.5,-2.6);
		
		\node[draw=none,fill=none] at (-5,-1.8) {$o_1 = ( \mathunderline{bough_4}{1},\ v_5,\ x(1)  )$} ;
		\node[draw=none,fill=none] at (-5,-2.2) {$o_2 = ( \mathunderline{bough_5}{2},\ v_2,\ x(2)  )$} ;
		\node[draw=none,fill=none] at (-5,-2.6) {$o_3 = ( \mathunderline{bough_1}{3},\ v_7,\ x(3)  )$} ;

		\Vertex[x=-1.35, y=-1.1, Math]{v_3}
		\Vertex[x=-0.4, y=-1.1, Math]{v_4}
		
		\Vertex[x=1.35, y=-1.1, Math]{v_6}
		
		\Vertex[x=3.15, y=-1.1, Math]{v_8}

		\begin{scope}
		\SetVertexNormal[FillColor  = pastel_5, LineColor=bough_5, LineWidth=3.5mm]
		\Vertex[x=0, y=2.8, Math]{}
		\end{scope}
		
		\begin{scope}
		\SetVertexNormal[FillColor  = pastel_4, LineColor=bough_4, LineWidth=1.75mm]
		\Vertex[x=1.5, y=1.4,Math]{}
		\Vertex[x=0, y=2.8, Math]{}
		\end{scope}
		
		\begin{scope}
		\SetVertexNormal[FillColor  = pastel_1]
		\Vertex[x=2.25, y=-1.1, Math]{v_7}
		\Vertex[x=2.55, y=0.2, Math]{b_7}
		\Vertex[x=1.5, y=1.4, Math]{b_3}
		\Vertex[x=0, y=2.8, Math]{b_1}
		\end{scope}

		\begin{scope}
		\SetVertexNormal[FillColor  = pastel_4]
		\Vertex[x=0.5, y=-1.1, Math]{v_5}
		\Vertex[x=0.8, y=0.2, Math]{b_6}
		\end{scope}

		\begin{scope}
		\SetVertexNormal[FillColor  = pastel_5]
		\Vertex[x=-2.25, y=-1.1, Math]{v_2}
		\Vertex[x=-2.55, y=0.2, Math]{b_4}
		\Vertex[x=-1.5, y=1.4, Math]{b_2}
		\end{scope}

		\Vertex[x=-3.15, y=-1.1, Math]{v_1}
		
		\Vertex[x=-0.8, y=0.2, Math]{b_5}

		\node[draw=none,fill=none, above=of b_1] {$H(b_1) = \{\mathunderline{bough_4}{1}, \mathunderline{bough_5}{2} , \mathunderline{bough_1}{3}\}$} ;
		\node[draw=none,fill=none, above left =of b_2] {$H(b_2) = \{ \mathunderline{bough_5}{2} \}$} ;
		\node[draw=none,fill=none, above left =of b_4] {$H(b_4) = \{ \mathunderline{bough_5}{2} \}$} ;
		\node[draw=none,fill=none, above right =of b_3] {$H(b_3) = \{ \mathunderline{bough_4}{1} , \mathunderline{bough_1}{3}\}$} ;
		\node[draw=none,fill=none, above right =of b_7] {$H(b_7) = \{ \mathunderline{bough_1}{3} \}$} ;
		
		\node[draw=none,fill=none, above right = 0 and 0 of v_8] {$H(v_8) = \{\}$} ;
		\node[draw=none,fill=none, above right = 0.4 and 0 of v_8] {$H(v_7) = \{ \mathunderline{bough_1}{3} \}$} ;
		
		\node[draw=none,fill=none, above left = 0.4 and 0.95 of v_2] {$H(v_1) = \{\}$} ;
		\node[draw=none,fill=none, above left = 0 and 0.8 of v_2] {$H(v_2) = \{ \mathunderline{bough_5}{2} \}$} ;
		
		\node[draw=none,fill=none, above left = 0  and -0.5 of b_6] {$H(b_6) = \{ \mathunderline{bough_4}{1} \}$} ;
		\node[draw=none,fill=none, above right = 0.4 and -0.3 of b_5] {$H(b_5) = \{\}$} ;

		\Edges[](b_1, b_2)
		\Edges[](b_1, b_3)
		\Edges[](b_4, b_2, b_5)
		\Edges(b_3, b_6)
		\Edges(b_3, b_7)
		
		\Edges(v_1, b_4, v_2)
		\Edges(v_3, b_5, v_4)
		\Edges(v_5, b_6, v_6)
		\Edges(v_7, b_7, v_8)
		
		\end{tikzpicture}
	}
	
	\caption{Example illustrating the definition of $H$. The set $H(b_i)$ keeps track of the indexes of the updates that are relevant at node $b_i$. These are exactly those that update a prefix that ends in a descendant of $b_i$. The Figure indicates the value of $H$ for all inner nodes and the leafs $v_1, v_2$, $v_7$ and $v_8$.}
	
	\label{fig:time-indexes}
	
\end{figure}

We continue with three important observations.

\begin{observation}\label{obs:one}
	$H(l)$ and $H(r)$ are disjoint, and $H(b)=H(l)\cup H(r)$. This implies that we can \emph{merge} the updates relevant to the children to obtain the updates relevant to $b$. Note that an update is relevant to $\log n + 1$ nodes of the tree, namely those along the path from the root of $B$ towards the last node in the list prefix that changes.
	See \Cref{fig:time-indexes} for an example.
\end{observation}

\begin{observation}\label{obs:two}
	The \Cref{obs:one} allows us leave out all the indices that do not occur in $H(b)$, as they are not relevant to $b$: 
	\[\Delta_i(b) = \Delta_0(b) + \left(\sum_{j\in H(b)} \phi_j (r) \right)- \left(\sum_{j\in H(b)} \phi_j (l) \right) \, , \text{ for any $i\in H(b)$.} \]
\end{observation}

\begin{observation}\label{obs:three}
In a bottom up traversal of the tree, we get $\phi_i (l)$ for all $i\in H(l)$ and $\phi_i (r)$ for all $i\in H(r)$. This means that we have some "missing" values for the sums in \Cref{obs:two}, namely the values of $\phi_i(r)$ for all $i \in H(l)$ and the values of $\phi_i(l)$ for all $i\in H(r)$. 
As explained for the sequential operation (see \Cref{sec:seqAddPrefix}), these values are trivial to get:
\begin{itemize}
	\item Updates $(i, v(i), x(i))$ with $i\in H(l)$ do not affect the descendants of $r$. Hence, for such an $i\in H(l)$ we have $\phi_i(r)=0$.
	\item Updates $(i, v(i), x(i))$ with $i\in H(r)$ change all descendants of $l$ by the same amount: for such an $i\in H(r)$ we have $\phi_i(l)=x(i)$. 
\end{itemize}
\end{observation}

The above Observations \ref{obs:one}-\ref{obs:three} suggest a parallel bottom-up procedure that produces for each node $b$:
\begin{itemize}
	\item The values of $H(b)$ as an array in sorted order. In the following the $i$-th largest value in $H(b)$ is denoted $H(b)[i]$. This means that $H(b)[i]$ is the index (w.r.t. the batch of updates) of the $i$-th update relevant at node $b$.
	
	\item An array $X(b)$ stores the increment of the updates relevant at $b$:
	\[X(b)[i] = x(j) \ \ \text{, where } j=H(b)[i] \ \ .\]
	This means that the $i$-th update relevant at $b$ is of the form \textsc{AddPath}($v(j)$, $X(b)[i]$) for some descendant $v(j)$ of $b$.
	
	\item An array $\Phi(b)$ storing $\phi_j(b)$ for the updates relevant at $b$. Specifically, we have \[\Phi(b)[i] = \phi_{j} (b) \ \ \text{, where } j=H(b)[i] \ \ .\]
\end{itemize}
Moreover, the procedure produces for each interior node $b$ the array $\Delta(b)$ containing the  intermediate states of the data structure at all relevant times: \[\Delta(b)[i]= \ \Delta_{j}(b) \ \ \text{, where } j=H(b)[i] \ \ .\]

In the following, we explain how to obtain these values for leafs, inner nodes, and the root of the binary tree $B$.

	\subsubsection{Leafs}
	At the leafs, we just need to group the updates by vertex and keep track of the relevant quantities. First, apply a stable sort by vertex to the operation sequence. Now, the operation sequence is sorted by vertex and tuples with the same vertex are sorted by time. From this, we can easily find which operations belong to a given vertex $v$ (by a binary search for $v$ and its successor). Then, for each vertex $v$ initialize $H(v)$ as the times of the tuples that have $v$ as a vertex and initialize $X(v)$ with the corresponding increments. Set $\Phi(v)$ to $X(v)$ - recall that $\Phi(v)$ records how much the minimum in the subtree changed for every operation relevant at $v$ and at the leafs the minimum changes by the increment.

	\subsubsection{Inner Nodes}
	 At an inner node $b$ with left child $l$ and right child $r$, we merge the results from its children, use prefix sums to generate $\Delta(b)$ in parallel, and construct $\Phi(b)$ based on $\Delta(b)$.

	Using \Cref{obs:one}, we merge the two sorted arrays $H(l)$ and $H(r)$ in parallel to receive $H(b)$ in sorted order. Similarly, we merge the update increments $X(l)$ and $X(r)$ to obtain $X(b)$ in sorted order.

	Using \Cref{obs:three}, we reconstruct the missing values for the right child $r$ and merge them with $\Phi(r)$ that we got from the right child. Proceed similarly for the left child $l$. Now, for all times $i$ relevant to $b$ (i.e. for all $i\in H(b)$), we have $\phi_i(l)$ and $\phi_i(r)$ each in an array sorted by increasing time $i$.

	Using \Cref{obs:two}, we construct $\Delta(b)$ as follows. We compute (in parallel) the all-prefix-sums over the array that contains the $\phi_i(l)$ for all $i$ in $H(b)$, and the array that contains the $\phi_i(r)$ for all $i$ in $H(b)$. Then, the observation immediately implies $\Delta_i(b)$ for all $i$ in $H(b)$ in parallel.
	
	Finally, we compute each entry of $\Phi(b)$ in parallel:
\begin{align*}
\Phi(b)[i] &=
	\begin{cases}
	\Phi(l)[i] \hfill & \text{\small if  $\Delta(b)[i-1] > 0$ and $\Delta(b)[i] > 0$}\enspace , \\
	\Phi(l)[i] - \Delta(b)[i-1] \hfill & \text{\small if $\Delta(b)[i-1] \leq 0$ and $\Delta(b)[i] > 0$}\enspace , \\
	\Phi(r)[i] \hfill & \text{\small if  $\Delta(b)[i-1] \leq 0$ and $\Delta(b)[i] \leq  0$}\enspace , \\
	\Phi(r)[i] + \Delta(b)[i-1] \hfill & \text{\small if $\Delta(b)[i-1] > 0$ and $\Delta(b)[i] \leq 0$} \enspace .
\end{cases}
\end{align*}

\subsubsection{Root} 
For the root $\rho$ of the binary tree, we proceed as for an internal node. Additionally, we generate the overall minimum weight after every update, which will be needed for the parallel \textsc{MinPath} queries. Observe that 
\begin{align*}
	\text{min}_i(\rho) = \text{min}_0(\rho) + \sum_{j=1}^{i} \phi_j(\rho) \enspace .
\end{align*}
Hence, we compute all the values $\text{min}_1(\rho), \dotsc, \text{min}_k(\rho)$ by a parallel all-prefix-sums computation on $\phi_1(\rho), \dotsc, \phi_k(\rho)$.

\subsubsection{Running Time}
By the proceeding described above, we obtain the following bounds for parallel \textsc{AddPrefix} operations.
\begin{lemma}
	Performing a batch of $k$ parallel \textsc{AddPrefix} operations on a list of length $n$ takes $O( k (\log n + \log k) + n)$ work and depth $O(\log n \log k)$ 
\end{lemma}
\begin{proof}
	At the leafs, we sort the operation sequence and perform $O(k)$ parallel binary searches. This takes $O( k \log k)$ work and $O(\log k)$ depth.
At every inner node $b$, we perform a constant number of parallel all-prefix-sums operations and parallel merge operations on arrays of size $O(H(b))$. By the bounds of merging sorted arrays~\cite{DBLP:journals/siamcomp/Cole88} and parallel prefix sums~\cite{Reif:1993:SPA:562546}, the work at an inner node $b$ is $O(H(b) + 1)$ and the depth is $O(\log (H(b)) + 1)$.

All nodes with the same distance to the root can be processed in parallel. Thus, the total work arising from nodes at distance $i$ to the root is $O(k + 2^i)$ and the depth is $O(\log k)$. Here, we used that summing $H(b)$ over all nodes at a fixed distance $i$ gives exactly $k$, because each leaf is the descendant of exactly one node at distance $i$.
Nodes with different distance from the root need to be processed by decreasing distance (bottom up). Hence, the overall work arising at inner nodes is $O(k \log n + n)$ and the overall depth is $O(\log n \log k)$.

At the root, we perform an additional parallel all-prefix-sums operation on an array of size $k$.
\end{proof}

\subsection{Parallel MinPrefix} \label{sec:parallel-minprefix}

The parallel update algorithm (batch of \textsc{AddPath} operations) produces all intermediate states for the data structure. If we store for each node in the data structure all the state it ever has (sorted by time), the value of a cell after the $i$-th update can be determined by doing a binary search on those states, taking $\Theta(\log k)$ time. Each query is then performed independently. Overall, this takes $\Theta( k \log n \log k)$ work and $\Theta(\log k \log n)$ depth, which would be $\Omega(\log k)$ more work compared to the original data structure.

To get rid of this logarithmic factor in work, we also perform the queries in a batch and we use parallel merging to avoid the binary searches. We obtain the following bounds:
\begin{lemma}\label{lem:parallel-minimum-path}
	Performing a batch of $k$ \textsc{MinPrefix} and \textsc{AddPrefix} operations on a list of length $n$ takes $O( k (\log n + \log k) + n)$ work and depth $O(\log n \log k)$.
\end{lemma}

The procedure is similar as for the updates. The queries are placed at the leafs and the nodes at the same distance from the root are processed in parallel. 

For a single query ($i$, $v(i)$), the following happens: the query gets processed bottom-up in all nodes on the path $P_i$ from the leaf $v(i)$ to the root, such that every internal node $b$ obtains an intermediate result $d_i(c)$ from its child $c$, updates this result based on $\Delta_i(b)$ to $d_i(b)$, and passes this result to its parent node.
Initially, leaf $v(i)$ sets $d_i{(v(i))} = 0$. Then, every internal node $b$ with left child $l$ and right child $r$ updates this result to
\begin{align*}
d_i(b) = \begin{cases}
d_i(l) & \text{if } \Delta_i(b) > 0,\\
d_i(r) & \text{if } v(i) \in r^{\downarrow} , \text{ and } d_i(r) + \Delta_i(b) < 0 ,\\
d_i(l)-\Delta_i(b) & \text{otherwise,}
\end{cases}
\end{align*}
where $d_i(l)$ and $d_i(r)$ have either been already computed, since the child lies on the path $P_i$, or are equal to zero.

To process all queries in parallel, every leaf now initiates an array that contains all its queries with the corresponding intermediate results sorted by time. Hence, an internal node $b$ obtains such an array $L$ from its left child and $R$ from its right child, which it merges (in parallel, by time) to an array $Q$ that now contains all queries and intermediate results relevant to $b$.

The last issue is to have the $\Delta$ values ready for each relevant query. In particular, we need the value that belongs to the last update before the query.

Thus, we record the time when this value was set by translating the indices of $b$'s relevant updates to the times that correspond to it: $\Delta(b)[i]$ is mapped to the tuple $(H(b)[i], \Delta(b)[i])$. 
Then, we merge $Q$ and the relevant $\Delta$-values in parallel and sorted by time. The new array contains a mix of queries and $\Delta$-values, sorted by time, such that now each query just needs to read the last $\Delta$-value on its left. 
This is achieved by a segmented broadcast, where each $\Delta$-value broadcasts its value to all following queries. A segmented broadcast can be implemented using a variant of the parallel all-prefix-sums algorithm~\cite{Reif:1993:SPA:562546}.

At the root, each query needs to read the overall minimum at the closest preceding time to the query. This can be achieved similarly using parallel merging and a segmented broadcast.

The overall runtime analysis is very similar to the one for the \textsc{AddPrefix}, since processing node $b$ has the same cost: it has work $O(H(b)+1)$ and depth $O(\log(H(b)) + 1)$.

\subsection{Tree Decomposition}\label{sec:tree-decomposition}

To solve \textsc{MinPath} and \textsc{AddPath} for general trees, we find a suitable decomposition of the tree $T$ into paths. This section presents a parallel algorithm to compute the decomposition we also used in the cache-oblivious minimum cut algorithm~\cite{DBLP:conf/ciac/GeissmannG17}. The idea is to repeatedly remove certain subpaths that start at a leaf, as follows: 

We call a path that starts at a leaf and ends at the first vertex that has a sibling on the way up to the root, a \emph{bough} of $T$. A bough ends at a vertex whose parent has multiple children or, if the tree $T$ is a path, at the root. See \Cref{fig:boughs-example} for an example.

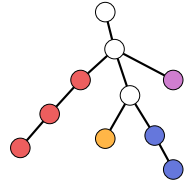
\begin{SCfigure}

		\resizebox{0.37\linewidth}{!} 
		{		
			\vspace{-1em}
			
			\resizebox{0.1\linewidth}{!} 
			{
			\begin{tikzpicture}
			
			\definecolor{mygray}{rgb}{0.8,0.8,0.8}

			\GraphInit[vstyle=Classic]

			\SetUpEdge[lw         = 1pt,
			color      = black,
			labelcolor = white,
			labeltext  = black,
			labelstyle = {sloped}]
			
			\begin{scope}
				\SetVertexNormal[FillColor  = bwsafe_1]
				\Vertex[x=0.4, y=1.2, L={$$}]{w_1}
				\Vertex[x=-0.6, y=0.1, L={$$}]{w_4}
				\Vertex[x=-1.55, y=-1, L={$$}]{w_6}			
			\end{scope}

			\begin{scope}
				\SetVertexNormal[FillColor  = bwsafe_2]
				\Vertex[x=1.2, y=-0.7, L={$$}]{w_7}
			\end{scope}
			
			\begin{scope}
				\SetVertexNormal[FillColor  = bwsafe_3]
				\Vertex[x=2.8, y=-0.6, L={$$}]{w_8}
				\Vertex[x=3.4, y=-1.7, L={$$}]{w_9}
			\end{scope}

			\begin{scope}
				\SetVertexNormal[FillColor  = white, TextColor=white]
				\Vertex[x=1.2, y=3.4, Math]{w_r}	
				\Vertex[x=1.5, y=2.2, Math]{w_0}	
				\Vertex[x=2, y=0.7, Math]{w_5}
			\end{scope}
			
			\begin{scope}
				\SetVertexNormal[FillColor  = bwsafe_4]
				\Vertex[x=3.4, y=1.2, L={$ $}]{w_3}			
			\end{scope}

			\begin{scope}
			\Edges[lw=2pt](w_0,w_3)0
			\Edges[lw=2pt](w_1, w_0)
			\Edges[lw=2pt](w_r, w_0, w_5)
			\Edges[lw=2pt](w_5, w_7)
			\Edges[lw=2pt](w_1,w_4,w_6)
			\Edges[lw=2pt](w_8, w_9)
			\Edges[lw=2pt](w_5, w_8)
			\end{scope}
			\end{tikzpicture}
		}
		}
	\caption{The tree above has $4$ boughs, indicated by the vertex colors. Each bough starts at a leaf and continues upwards until reaching the first node that has a sibling.}
	\label{fig:boughs-example}
	\vspace{-1em}
\end{SCfigure}

The algorithm repeats the following until no edges remain:
\begin{enumerate}
	\item Identify the boughs of $T$. Each bough is in the decomposition.
	\item Remove all vertices that are part of a bough.
\end{enumerate}
Note that shrinking boughs is also used for other parallel graph problems~\cite{DBLP:journals/jal/Ramachandran97}. The algorithm is also related to a parallel tree contraction algorithm~\cite{DBLP:conf/focs/MillerR85}.
We obtain the following bounds: 

\begin{lemma}\label{lem:tree-decomposition}
	A tree $T$ with $n$ vertices can be decomposed into a set of pairwise vertex-disjoint paths $P$ such that: 
	\begin{itemize}[leftmargin=2em]
		\item Each root-to-leaf path in $T$ intersects at most $\log_2{n} $ paths in $P$.
		\item This takes work $O(n \log n)$ and depth $O(\log^2 n)$ (Las Vegas).
	\end{itemize}
\end{lemma}

Observe that, in each repetition, the number of leaves is at least halved. This is because each leaf in the new tree had at least two children before the contraction. Hence, there are at most $\log_2 n$ repetitions. This implies that every root-to-leaf path in $T$ is decomposed into at most $\log_2 n$ paths.

We give a randomized algorithm to find the boughs in parallel in \Cref{sec:finding-the-boughs}. With high probability, it has work $O(n)$ and depth $O(\log n)$. Removing the vertices takes work $O(n)$ and depth $O(1)$. As the tree decomposition has $O(\log n)$ iterations, \Cref{lem:tree-decomposition} follows.

\subsubsection{Finding the Boughs}\label{sec:finding-the-boughs}

Next, we show how to find the boughs in parallel, implying the work and depth bounds from \Cref{lem:parallel-boughs}. Observe that the boughs are induced subpaths and thus we can use ideas from parallel list ranking~\cite{DBLP:journals/algorithmica/AndersonM91}. 

We call a vertex in a tree a \textit{branching vertex} if it has at least two children, and similarly we say that a vertex is \emph{non-branching} if it has at most one child.
\begin{enumerate}
	\item As long as there are non-branching nodes in $T$, find an independent set of edges (i.e. edges that do not share endpoints) whose endpoints are both non-branching and contract this set of edges. When merging two vertices, keep track of the original labels of the vertices that were merged into that vertex. Specifically, keep the labels as linked lists with head and tail pointers.
	\item After the procedure converges, the leaf vertices contain the labels of the boughs (as a linked list). 
\end{enumerate}

\begin{lemma}\label{lem:parallel-boughs}
	The boughs of a tree with $n$ vertices can be identified with $O(n)$ work and $O(\log n)$ depth (Las Vegas randomized).
\end{lemma}
\begin{proof}
	We can find large independent sets with $O(N)$ work and $O(1)$ depth (where $N$ is the number of non-branching nodes), for example using the random-mate technique introduced for list ranking~\cite{DBLP:journals/algorithmica/AndersonM91,DBLP:conf/ipps/ArgeGS10}. This ensures that, with high probability, the number of non-branching internal vertices decreases by a constant factor at each repetition and $O(\log n)$ repetitions suffice until all internal vertices are branching. Merging two non-branching vertices takes $O(1)$ work because they have constant degree. As there can be at most $n-1$ merge operations, the overall work from merging is $O(n)$. The depth to contract an independent set of edges connecting non-branching nodes is $O(1)$, as each such edge can be contracted completely parallel to the other edges and needs only $O(1)$ pointers to change.
\end{proof}

It is possible to make the algorithm deterministic by replacing the randomized independent set construction by a deterministic one: Construct a $3$-coloring of the tree and choose the color $c$ with the largest number of non-branching internal vertices. For each internal non-branching vertex of that color, add the edge connecting it to its child to the independent set.

A $3$-coloring of a tree is constructed deterministically in depth $O(\log ^ {*} n)$ and work $O(n \log ^ {*} n)$ \cite{DBLP:journals/jacm/GoldbergT88}. Using this deterministic approach, the work to decompose the tree as in \Cref{lem:tree-decomposition} increases to $O(n \log^2 n \log ^ {*} n)$ and the depth increases to $O(\log^2 n \log ^ {*} n)$.

\subsection{Parallel MinPath and AddPath}\label{sec:minpath-addpath}

Each \textsc{MinPath} and \textsc{AddPath} operation corresponds to $O(\log n)$ \textsc{MinPrefix} and \textsc{AddPrefix} operations, respectively, which can be processed in parallel. For the \textsc{MinPath} operations, the smallest result of the $O(\log n)$ \textsc{MinPrefix} queries can be found sequentially after they have completed. We conclude:

\begin{lemma}
	Performing a batch of $k$ \textsc{MinPath} and \mbox{\textsc{AddPath}} operations on a tree with $n$ nodes takes $O( k \log n (\log n + \log k) + n \log n)$ work and depth $O(\log n (\log k + \log n))$ (Las Vegas randomized).
\end{lemma}

\section{Parallel Minimum Cuts}\label{sec:parallel-mincuts}

The parallel \emph{Minimum Path} structure is the missing puzzle-piece in a parallelization of Karger's algorithm. This solved, it remains to show how to create the batch of Minimum Path operations and how to combine the results of the Minimum Path structure into a minimum cut, achieving the following overall bounds.

\begin{theorem}\label{thm:mincut-in-parallel}
	The minimum cut of a graph can be computed with depth $O(\log ^ 3 n)$ and work $O(m \log ^ 4 n)$ (Monte Carlo randomized).
\end{theorem}

In the following, we consider a (rooted) spanning tree $T$ of the input graph $G$ and we want to compute the smallest cut that cuts at most two edges of $T$. We will do so with work $O(m \log ^ 3 n)$ and depth $O(\log ^2 n)$. Together with \Cref{lem:sampling} this implies the main result.

Karger already showed a parallel algorithm that computes the smallest cut that cuts \emph{exactly one} edge of a given spanning tree. In fact, the algorithm computes for each vertex $v$, the value of the cut $v^{\downarrow}$ that has the descendants of $v$ on one side of the cut (and therefore cuts only the edge from $v$ to its parent in $T$).

\begin{lemma}[Karger~\cite{karger2000minimum}]\label{lem:one-edge-cut}
	A smallest cut that cuts exactly one edge of a given spanning tree can be computed with work $O(m)$ and depth $O(\log m)$.
\end{lemma}

We therefore focus on giving a parallel algorithm for the case where \emph{exactly two} edges of the spanning tree are cut.

\subsection{Cutting two edges of a spanning tree}\label{sec:cutting-two-edges}

Our parallel algorithm uses ideas from Karger's sequential algorithm~\cite{DBLP:journals/siamrev/Karger01} to reduce the problem to a set of Minimum Path operations, which we already showed how to perform in parallel in \Cref{sec:parallel-minpath}. 

We are given a spanning tree $T$ of $G$. Assume that the smallest cut $C$ that cuts at most two edges of $G$ cuts the edges $(u, v)$ and $(s, t)$ of $T$ (where $u$ is the parent of $v$ and $s$ the parent of $t$). Let us focus on the case where $u$ is not an ancestor of $s$ and vice versa. See \Cref{fig:two-edges-cuts} for an example. The case where $u$ is an ancestor of $s$ is similar (see \Cref{sec:comparable-vertices}).

If we take the value of the cut $t^{\downarrow}$ and add the value of the cut $v^{\downarrow}$ we incorrectly count (twice) the edges that go between the descendants of $v$ and the descendants of $t$. We will use the Minimum Path structure to keep track of these "extra" edges that go between the two parts of the tree.

In the following, we will explain how the algorithm handles possibilities for these four vertices, including for instance if one of $v$ or $t$ was a leaf. Since the algorithm does not know which two edges to cut is best, all possibilities need to be considered. 

\pagebreak

\subsubsection{Handling a leaf}
Assume that $v$ is a leaf of $T$ (see \Cref{fig:two-edges-cuts} for an illustration). The algorithm proceeds as follows:

\begin{enumerate}
	\item 
Initialize a Minimum Path Structure where the initial weight of each vertex $x$ is given by the value of the cut $x^{\downarrow}$. Recall that these values are computed by the algorithm from \Cref{lem:one-edge-cut}.
	\item Then, for each edge $e=(v, y)$ incident to $v$, perform \\ \textsc{AddPath}($v$, $-2w(e)$).
\end{enumerate} 

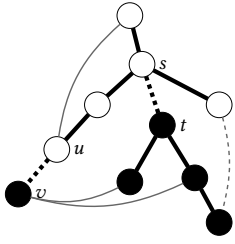
\begin{SCfigure} 
	\centering
	\resizebox{0.38\linewidth}{!} {
		
		\begin{tikzpicture}
		
		\GraphInit[vstyle=Classic]
		
		\SetUpEdge[lw         = 1pt,
		color      = black,
		labelcolor = white,
		labeltext  = black,
		labelstyle = {sloped}]
		
		\begin{scope}
		\SetVertexNormal[FillColor  = black]
		\Vertex[x=-1.55, y=-1,L={\huge$v$}]{w_6}	
		\Vertex[x=1.2, y=-0.7,L={$$}]{w_7}	
		\Vertex[x=2.8, y=-0.6,L={$$}]{w_8}
		\Vertex[x=3.4, y=-1.7,L={$$}]{w_9}
		\Vertex[x=2, y=0.7, L={\huge$t$}]{w_5}
		\end{scope}
		
		\begin{scope}
		\SetVertexNormal[FillColor  = white]
		\Vertex[x=0.4, y=1.2, L={$$}]{w_1}
		\Vertex[x=-0.6, y=0.1, L={\huge$u$}]{w_4}
		\Vertex[x=1.2, y=3.4, L={$$}]{w_r}	
		\Vertex[x=1.5, y=2.2, L={\huge$s$}]{w_0}	
		\Vertex[x=3.4, y=1.2,L={$$}]{w_3}	
		\end{scope}
		
		\begin{scope}[]
		\Edges[lw=3pt](w_0,w_3)
		\Edges[lw=3pt](w_1, w_0)
		\Edges[lw=3pt](w_r, w_0)
		\Edges[lw=3pt](w_5, w_7)
		\Edges[lw=3pt](w_1, w_4)
		\Edges[lw=3pt](w_8, w_9)
		\Edges[lw=3pt](w_5, w_8)
		\end{scope}
		
		\begin{scope}[style=dashed]
		\Edges[lw=3pt](w_0, w_5)
		\Edges[lw=3pt](w_4,w_6)
		\end{scope}

		\begin{scope}
		\tikzstyle{EdgeStyle}=[color=black!60 ,bend left=20]
		\Edges[](w_8, w_6)
		\Edges[](w_7, w_6)
		\Edges[](w_4, w_r)
		\end{scope}
		
		\begin{scope}
		\tikzstyle{EdgeStyle}=[color=black!60 ,bend left=15, style=dashed]
		
		\Edges[](w_3, w_9)
		\end{scope}

		\end{tikzpicture}
	}
	\caption{The weight of the (black, white) cut that cuts $(u, v)$ and $(s, t)$ is given by the dashed edges. This equals the edges with at least one black endpoint minus the edges with two black endpoints. }
	\label{fig:two-edges-cuts}
	
\end{SCfigure}

Observe that after these steps, the weight in the minimum path structure of a vertex $x$ that is not an ancestor $v$ is exactly the value of the cut $x^{\downarrow}$ minus twice the weight of the edges between $v$ and the descendants of $x$.
Moreover, since we assume that $C$ is a smallest cut of $G$ among those that cut at most two edges of $T$, at least one of $v$'s neighbors has to be a descendant of $t$. Otherwise, the cut $t^{\downarrow}$ would be smaller than the cut $C$, which leads to a contradiction.
This implies that the value of $C$ is given by the value of the cut $v^{\downarrow}$ plus the minimum weight of a node $x$, such that $x$ is a neighbor but not an ancestor of $v$.
Thus, we  find the value of the cut $C$ as follows:
\begin{enumerate}
	\item Add $\infty$ to the weight of all ancestors of $v$ by \textsc{AddPath}($v$, $\infty$).
	\item For each neighbor $x$ of $v$, call \textsc{MinPath}($x$) and keep the smallest result.
	\item Add the value of the cut $v^{\downarrow}$. This is the value of the cut $C$.
\end{enumerate}

\subsubsection{Handling a bough} 
The observations from before can be generalized to also handle boughs. Recall from \Cref{sec:finding-the-boughs} that a bough is a maximal induced subpath that contains a leaf. 
Similar to the case where $v$ is a leaf, we use the Minimum Path Structure to keep track of the "extra" edges that go between $v^{\downarrow}$ and $t^{\downarrow}$. 
The following procedure handles the case where $v$ is in a bough.

\begin{enumerate}
	\item Initialize the minimum path structure just as in the leaf case.
	\item Start at the leaf of the bough and walk up the bough. At every node $y$ in the bough:
	\begin{enumerate}
	\item If $y$ is a leaf, perform \textsc{AddPath}($y$, $\infty$) .
	\item Moreover, for every edge $e=(y, x)$ incident to $y$, perform \textsc{AddPath}($x$, $-2w(e)$).
	\item Afterwards, for every neighbor $x$ of $y$, perform the query MinPath($x$). Record the smallest result.
	\end{enumerate}
\end{enumerate}

Consider the state of the minimum path structure when the above procedure  has processed some node $y$ in the bough. Consider another node $x$ that is not an ancestor of any node in the bough. Then, the weight of $x$ is equal to the value of the cut $x^{\downarrow}$ minus twice the weight of the "extra" edges that exist between $y^\downarrow$ and $x^\downarrow$.

Moreover, by the minimum assumption of $C$, it must hold that $v$ has a descendant that is a neighbor of a descendant of $t$.

Putting these two observations together, we conclude that it is indeed enough to perform a \textsc{MinPath}($y$) query for every neighbor $y$ of a node in the bough and record the smallest result. When we have processed $v$, the smallest result seen so far plus the value of the cut $v^{\downarrow}$ gives the value of the cut $C$.

\subsubsection{Handling a general tree} 
To handle a general subtree we repeat the procedure for every bough. Then, we contract all edges (of the spanning tree and the overall graph) with at least one endpoint in a bough. Afterwards, we recurse in the new tree. We call such a phase a {\em bough-phase}. Note that this gives the same decomposition of the tree as in \Cref{sec:tree-decomposition}.

In each bough-phase, the boughs can be handled in any order. However, after handling a bough, we need to restore the weights to their initial state. This is done by reversing the order and the sign of all \textsc{AddPath} operations: We visit the nodes in the bough top-down and replace each AddPath($x$, $w$) by AddPath($x$, $-w$).
In the end, we return the smallest cut value found.

\subsection{Generating the batch of Minimum Path operations}\label{sec:generating-the-batch}

We show how to generate the batch of minimum path operations for one bough-phase. We already saw how to identify the boughs in \Cref{lem:parallel-boughs}. The remaining difficulty is to compute (in parallel) the order in which the edges are accessed. 

Observe that each edge is accessed at most four times: for each of its endpoints once on the way going upwards in the bough containing this endpoint, and once on the way down. We get the order and the operations as follows.

\begin{enumerate}
	\item Order each bough by list-ranking and give each leaf a unique identifier. Then, the order in which a vertex is visited is derived from the number of the leaf of its bough and its position in the list-ranking. See \Cref{fig:bough-vertex-order} for an example. 
	\item Each leaf creates an \textsc{AddPath}($v$, $\infty$) at the time of its first visit and an \textsc{AddPath}($v$, $-\infty$) at the time of its second visit.
	\item When a node $y$ in a bough is visited at times $t_1$ and $t_2$, then every neighbor $x$ of $y$ creates an update that corresponds to \textsc{AddPath}($x$, $-2w(e)$) at time $t_1$ and a query that corresponds to \text{MinPath}($x$) at time $t_1$. Moreover, it creates an update that corresponds to \textsc{AddPath}($x$, $2w(e)$) at time $t_2$ (this undoes the operation of the former update). Each edge in a graph is thus accessed at most four times, namely every time one of its endpoints is visited. See \Cref{fig:bough-edge-order} for an illustration.
\item The queries and updates are sorted according to the visit times, where additionally operations with the same visit time are ordered such that updates come before queries. This gives the operation sequence that handles all the boughs in the tree. 
\end{enumerate}

\begin{figure}[t]
	\begin{minipage}[t]{.48\linewidth}
		\centering	\resizebox{0.85\linewidth}{!} {
			
			\begin{tikzpicture}

			\definecolor{mygray}{rgb}{0.8,0.8,0.8}
			
			\GraphInit[vstyle=Classic]

			\SetUpEdge[lw         = 1pt,
			color      = black,
			labelcolor = white,
			labeltext  = black,
			labelstyle = {sloped}]
			
			\begin{scope}
			\SetVertexNormal[FillColor  = bwsafe_1]
			\Vertex[x=0.4, y=1.2, L={\huge$3,4$}]{w_1}
			\Vertex[x=-0.6, y=0.1, L={\huge$2,5$}]{w_4}
			\Vertex[x=-1.55, y=-1, L={\huge$1,6$}]{w_6}			
			\end{scope}
			
			\begin{scope}
			\SetVertexNormal[FillColor  = bwsafe_2]
			\Vertex[x=1.2, y=-0.7, L={\huge$7,8$}]{w_7}
			\end{scope}
			
			\begin{scope}
			\SetVertexNormal[FillColor  = bwsafe_3]
			\Vertex[x=2.8, y=-0.6, L={\huge$10,11$}]{w_8}
			\Vertex[x=3.4, y=-1.7, L={\huge$9,12$}]{w_9}
			\end{scope}

			\begin{scope}
			\SetVertexNormal[FillColor  = white, TextColor=white]
			\Vertex[x=1.2, y=3.4, Math]{w_r}	
			\Vertex[x=1.5, y=2.2, Math]{w_0}	
			\Vertex[x=2, y=0.7, Math]{w_5}
			\end{scope}
			
			\begin{scope}
			\SetVertexNormal[FillColor  = bwsafe_4]
			\Vertex[x=3.4, y=1.2, L={\huge$13,14$}]{w_3}			
			\end{scope}
			
			\begin{scope}
			\Edges[lw=2pt](w_0,w_3)
			\Edges[lw=2pt](w_1, w_0)
			\Edges[lw=2pt](w_r, w_0, w_5)
			\Edges[lw=2pt](w_5, w_7)
			\Edges[lw=2pt](w_1,w_4,w_6)
			\Edges[lw=2pt](w_8, w_9)
			\Edges[lw=2pt](w_5, w_8)
			\end{scope}
			
			\end{tikzpicture}
		}
		\caption{The boughs of the tree are indicated by colors. The node labels indicate a possible order in which they are visited by the algorithm. Each node in a bough is visited twice. Once on a bottom-up traversal of its bough and then on a top-down traversal of its bough. The order between boughs is arbitrary.}
		\label{fig:bough-vertex-order}
	\end{minipage}
	\hfill
	\begin{minipage}[t]{.48\linewidth}
		\centering
		\resizebox{0.85\linewidth}{!} {
			
			\begin{tikzpicture}
			
			\GraphInit[vstyle=Hasse]
			
			\SetUpEdge[lw         = 1pt,
			color      = black,
			labelcolor = white,
			labeltext  = black,
			labelstyle = {sloped}]
			
			\begin{scope}
			\SetVertexNormal[FillColor  = bwsafe_1]
			\Vertex[x=0.4, y=1.2,Math]{w_1}
			\Vertex[x=-0.6, y=0.1,Math]{w_4}
			\Vertex[x=-1.55, y=-1,Math]{w_6}			
			\end{scope}
			
			\begin{scope}
			\SetVertexNormal[FillColor  = bwsafe_2]
			\Vertex[x=1.2, y=-0.7,Math]{w_7}
			\end{scope}
			
			\begin{scope}
			\SetVertexNormal[FillColor  = bwsafe_3]
			\Vertex[x=2.8, y=-0.6,Math]{w_8}
			\Vertex[x=3.4, y=-1.7,Math]{w_9}
			\end{scope}

			\begin{scope}
			\SetVertexNormal[FillColor  = white, TextColor=white]
			\Vertex[x=1.2, y=3.4, Math]{w_r}	
			\Vertex[x=1.5, y=2.2, Math]{w_0}	
			\Vertex[x=2, y=0.7, Math]{w_5}
			\end{scope}
			
			\begin{scope}
			\SetVertexNormal[FillColor  = bwsafe_4]
			\Vertex[x=3.4, y=1.2,Math]{w_3}			
			\end{scope}

			\begin{scope}
			\Edges[lw=2pt](w_0,w_3)
			\Edges[lw=2pt](w_1, w_0)
			\Edges[lw=2pt](w_r, w_0, w_5)
			\Edges[lw=2pt](w_5, w_7)
			\Edges[lw=2pt](w_1,w_4,w_6)
			\Edges[lw=2pt](w_8, w_9)
			\Edges[lw=2pt](w_5, w_8)
			\end{scope}
			
			\begin{scope}
			\tikzstyle{EdgeStyle}=[color=black!60 ,bend left=15]
			\tikzset{LabelStyle/.style = {above=0.15cm}}
			\Edges[label= {\huge$1,6,7,8$},labelstyle={right}](w_7, w_6)
			\Edges[label= {\huge$2,5$},labelstyle={right}](w_4, w_r)
			\Edges[label= {\huge$9,12,13,14$},labelstyle={right}](w_3, w_9)
			\end{scope}

			\end{tikzpicture}
		}
		\caption{The order in which the edges are visited is given by the order of the vertices. An edge is visited whenever one of its endpoints is visited. Thus, each edge is visited two or four times. The figure indicates the times when the non-tree edges (grey) are visited, based on the example on the left.}
		\label{fig:bough-edge-order}
	\end{minipage}

\end{figure}

\begin{lemma}\label{lem:parallel-batch}
	Generating the batch of Minimum Path operations from \Cref{sec:cutting-two-edges} that handles all the boughs of a tree takes work $O(m\log n)$ and depth $O(\log n)$.
	\end{lemma}
	\begin{proof}
		List ranking of the boughs takes $O(\log n)$ depth and $O(m)$ work~\cite{DBLP:journals/algorithmica/AndersonM91}. The vertices derive their visit time with $O(\log n)$ depth and $O(n)$ work, the edges can then be processed completely in parallel with $O(1)$ depth and $O(m)$ work. Finally, sorting the batch of $O(m)$ operations takes $O(\log n)$ depth and $O(m \log n)$ work~\cite{DBLP:journals/siamcomp/Cole88}.
	\end{proof}

\subsection{Extracting the Minimum Cut.} \label{sec:extracting-the-cut}
To finally find the smallest cut of $G$ that cuts at most two edges of a given spanning tree $T$, we need to repeatedly apply such batches of operations in every bough-phase:

\begin{enumerate}
	\item 	For all vertices $x$, compute the value of the cut $x^{\downarrow}$ of $G$ that has the descendants of $x$ in $T$ on one side of the cut (using the algorithm from \Cref{lem:one-edge-cut}). In particular, this yields the smallest cut that cuts exactly one edge of $T$ . 
	\item Find the boughs of $T$. Contract all edges incident to a node in a bough (contract the edges in $T$ and $G$ at the same time) and recurse until the graph has a single vertex. This generates a sequence of graphs $G_1, G_2, \dotsc, G_k$ and corresponding spanning trees $T_1, T_2, \dotsc, T_k$. Note that $G_1=G$ and $T_1=T$.
	\item In parallel, for each of the pairs $T_i$ and $G_i$ in the sequence, run the algorithm from \Cref{sec:generating-the-batch} to generate the necessary Minimum Path operations.
	\item 		For each of the trees $T_i$, execute the batch of Minimum Path operations in parallel (\Cref{sec:minpath-addpath}). Find the smallest value $\text{MinPath}(v)$ + $C(v^{\downarrow})$ over all vertices $v$ (and all queries to $v$), where $C(v^{\downarrow})$ is the value of the cut of $G$ that has the descendants of $v$ in $T_i$ on one side of the value. These values have already been computed when finding the smallest cut that cuts one edge of $T$.
	\item The smallest value $\text{MinPath}(v)$ + $C(v^{\downarrow})$ found overall is the value of the smallest cut that cuts exactly two edges of $T$.	
\end{enumerate}

The algorithm obtains the following bounds:

\begin{lemma}\label{lem:final}
	Finding the smallest cut of a graph $G$ that cuts at most two edges of a given spanning tree $T$ of $G$ takes work $O(m  \log^3 n)$ and depth $O(\log ^2  n)$ (Las Vegas randomized).
\end{lemma}
\begin{proof}	
	 Using \Cref{lem:parallel-boughs}, generating the graphs $G_1, G_2, \dotsc, G_k$ and corresponding spanning trees $T_1, T_2, \dotsc, T_k$ takes $O(m \log n)$ work and $O(\log^2 n)$ depth. Since the number of leaves is at least halved in each bough-phase, $k=O(\log n)$. Contracting the edges in the bough is just a matter of mapping each vertex in a bough to the vertex that is the parent of the topmost vertex in the bough, look up the value of the endpoints of each vertex in the mapping and remove loops. It is not necessary to combine parallel edges.	
	
	For each graph in $G_1, G_2, \dotsc, G_k$, generating the queries takes $O(m \log n)$ work and $O(\log n)$ depth (by \Cref{lem:parallel-batch}). Executing the batch of $O(m)$ minimum path operations and finding the minimum result takes $O(m \log ^2 n)$ work and $O(\log^2 n)$ depth (by \Cref{lem:parallel-minimum-path}).
	As this can be done in parallel for each of the $O(\log n)$ graphs, the algorithm takes $O(m \log^3 n)$ work and $O(\log^2 n)$ depth overall.
		
	The algorithm can be easily adapted to also output the edges of $T$ that define the cut, essentially by recording the edges that generated a Minimum Path query.
	\end{proof}
	
	Together with \Cref{lem:sampling}, \Cref{lem:final} implies our main result as stated in \Cref{thm:mincut-in-parallel}.

\section{Cache-oblivious algorithm}\label{sec:CO}

The parallel minimum cut algorithm gives improved cache miss bounds to compute a minimum cut. The key difference to our previous cache-oblivious algorithm~\cite{DBLP:conf/ciac/GeissmannG17} is the implementation of the minimum path structure. The parallel minimum path algorithm from \Cref{sec:parallel-minpath} uses operations (such as merging sorted lists and computing prefix sums) that are easily made cache-efficient. 

The cache-oblivious model~\cite{Frigo99cache-obliviousalgorithms,Brodal} considers a fully-associative cache with optimal replacement strategy of size $M$ with cache lines of width $B$. The parameters $B$ and $M$ of the machine cannot be used in the algorithm description, hence the name cache-\emph{oblivious}.

We obtain the following bounds to compute a minimum cut if we replace the minimum path structure in our previous cache-oblivious algorithm~\cite{DBLP:conf/ciac/GeissmannG17} with the the data structure from \Cref{sec:parallel-minpath}.

\begin{theorem}
	Finding a minimum cut incurs $O\left(\frac{m \log^ 4 n}{ B} + 1\right)$ cache misses and takes $O(m \log ^ 4 n)$ time (Monte Carlo randomized).
\end{theorem}

\section{Conclusion}\label{sec:conclusion}

Compared to the best sequential algorithm, our algorithm performs only $O(\log n)$ more work, namely $O(m \log ^ 4 n)$ work and $O(\log ^3 n)$ depth.
It remains an open problem to find a work-optimal minimum cut algorithm that has poly-logarithmic depth.

The $\Omega(\log ^ 3 n)$ depth of our algorithm comes from the algorithm to find the suitable spanning trees (where the minimum cut crosses at most two edges of one of the spanning trees). The rest of our algorithm has depth $O(\log ^2 n)$. Consequently, a lower depth algorithm to find a suitable spanning tree would yield a lower depth minimum cut algorithm.

\clearpage
\appendix

\section{When the minimum cut is the difference of two descendant sets}\label{sec:comparable-vertices}

We discussed the situation when the minimum cut is given by the union of the descendant sets of two vertices in in \Cref{sec:cutting-two-edges}. However, it is also possible that the minimum cut equals the \emph{difference} of the descendants of two vertices. See \Cref{fig:two-edges-cuts-comparable} for an illustration of how such a cut is structured. The reduction onto Minimum Path operations follows Karger \cite{karger2000minimum}, we include it for completeness. The bounds match those obtained for \Cref{lem:final}.

The algorithm walks along the boughs of the tree in the same way as in \Cref{sec:cutting-two-edges}.
It uses the minimum path structure to keep track of the edges with an endpoint in $v^{\downarrow}$ and one endpoint in $t^{\downarrow}$ for all $t$ whilst iterating over $v$:

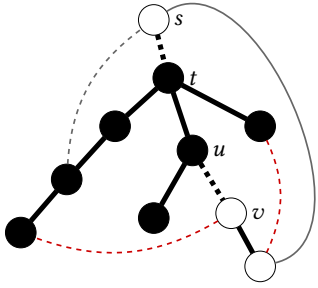
\begin{figure} 
	\centering
	\resizebox{0.6\linewidth}{!} {
		
		\begin{tikzpicture}
		
		\GraphInit[vstyle=classic]
		
		\SetUpEdge[lw         = 1pt,
		color      = black,
		labelcolor = white,
		labeltext  = black,
		labelstyle = {sloped}]
		
		\begin{scope}
		\SetVertexNormal[FillColor  = black]
		\Vertex[x=-1.55, y=-1,L={$$}]{w_6}	
		\Vertex[x=1.2, y=-0.7,L={$$}]{w_7}	
\Vertex[x=0.4, y=1.2, L={$$}]{w_1}
		\Vertex[x=2, y=0.7, L={\huge$u$}]{w_5}
		\Vertex[x=-0.6, y=0.1, L={$$}]{w_4}
		\Vertex[x=1.5, y=2.2, L={\huge$t$}]{w_0}
		\Vertex[x=3.4, y=1.2,L={$$}]{w_3}
		\end{scope}
		
		\begin{scope}
		\SetVertexNormal[FillColor  = white]

		\Vertex[x=1.2, y=3.4, L={\huge$s$}]{w_r}	
		\Vertex[x=2.8, y=-0.6,L={\huge$v$}]{w_8}
		\Vertex[x=3.4, y=-1.7,L={$$}]{w_9}
		\end{scope}
		
		\begin{scope}[]
		\Edges[lw=3pt](w_0,w_3)
		\Edges[lw=3pt](w_1, w_0)
		
		\Edges[lw=3pt](w_5, w_7)
		\Edges[lw=3pt](w_1, w_4)
		
		\Edges[lw=3pt](w_4,w_6)
		\Edges[lw=3pt](w_0, w_5)
		\Edges[lw=3pt](w_8, w_9)
		\end{scope}
		
		\begin{scope}[style=dashed]
		\Edges[lw=3pt](w_r, w_0)
		\Edges[lw=3pt](w_5, w_8)
		
		\end{scope}
		
		\begin{scope}
		\tikzstyle{EdgeStyle}=[color=red!80!black]
		
		\end{scope}
		
		\begin{scope}
		\tikzstyle{EdgeStyle}=[color=black!60 ,bend left=100]
		\Edges[](w_r, w_9)
		\end{scope}
		
		\begin{scope}
		\tikzstyle{EdgeStyle}=[color=black!60 ,bend left=25, style=dashed]
		\Edges[](w_4, w_r)
		\end{scope}
		
		\begin{scope}
		\tikzstyle{EdgeStyle}=[color=red!80!black ,bend left=25, style=dashed]
		\Edges[](w_8, w_6)
		\Edges[](w_3, w_9)
		\end{scope}

		\end{tikzpicture}
	}
	\caption{The weight of the (black, white) cut that cuts $(u, v)$ and $(s, t)$ is given by the dashed edges.  The idea to compute the value of this cut is to start with the value of the cut $t^{\downarrow}$ minus the value of the cut $v^{\downarrow}$ and think about which edges have been incorrectly counted. Those are the red edges, which have one endpoint in $v^{\downarrow}$ and one endpoint in $t^{\downarrow}-v^{\downarrow}$: They have been incorrectly subtracted and we need to add the weight of those edges back (twice). This weight is related to the total weight of the edges with one endpoint in $v^{\downarrow}$ and one endpoint in $t^{\downarrow}$, from which we still need to subtract the weight of the edges with both endpoints in $v^{\downarrow}$. }
	\label{fig:two-edges-cuts-comparable}
\end{figure}


Initialize the minimum path structure such that the weight of each vertex $v$ is the weight of the cut $v^{\downarrow}$. Traverse the vertexes in the same order as in \Cref{sec:cutting-two-edges}.
\begin{enumerate}
	\item At vertex $v$, for each neighbor $u$, call \textsc{AddPath}($u$, $2w(u, v)$).
	\item Then, call \textsc{MinPath}($v$).
\end{enumerate}
When walking back down the bough, undo the \textsc{AddPath} operations by reversing the sign. After all boughs are processed, contract the boughs and proceed recursively.

It remains to compute for each vertex $v$ the weight of the edges with both endpoints descending from $v$. This value needs to be subtracted from the result of the \textsc{MinPath}($v$) query.

Compute for each edge $(u, v)$ in the graph $G$ the lowest common ancestor \cite{DBLP:journals/siamcomp/SchieberV88} of $u$ and $v$ in the spanning tree $T$. Sum up for each vertex $v$ the total weight of the edges whose lowest common ancestor is $v$. Then, sum up these values to obtain for each vertex the total weight of the edges with both endpoints in $v^{\downarrow}$. Denote this value by $\rho^{\downarrow}(v)$. This final summation can be done by decomposing into paths as in \Cref{sec:tree-decomposition} and performing parallel all-prefix-sums.  

Finally, add to each \textsc{MinPath}($v$) result the precomputed value $4\rho^{\downarrow}(u)$ and the weight of the cut $v^{\downarrow}$. The smallest such result is the weight of the desired cut.




\bibliographystyle{ACM-Reference-Format}
\balance
\bibliography{parallel-mincuts}

\end{document}